\documentclass[10pt, twocolumn, comsoc]{IEEEtran}

\usepackage{graphicx}
\usepackage[noadjust]{cite}
\usepackage{mcite}
\usepackage{amsfonts,helvet}
\usepackage{fancyhdr}
\usepackage{threeparttable}
\usepackage{booktabs}
\usepackage{amsthm}
\usepackage{siunitx}
\usepackage{amssymb}
\usepackage{dsfont}
\usepackage{xcolor}
\usepackage{eucal}
\usepackage{amsmath}
\usepackage[ruled,vlined]{algorithm2e}
\usepackage{enumerate}
\usepackage{cancel}
\usepackage{comment}
\usepackage{algpseudocode}
\usepackage{bm}
\usepackage{epstopdf}
\usepackage{placeins}

\newtheorem{definition}{Definition}

\newtheorem{lemma}{Lemma}

\newtheorem{proposition}{Proposition}
\theoremstyle{definition}
\newtheorem{remark}{Remark}

\SetKwInput{KwInput}{Initialize}
\SetKwInput{KwOutput}{Output}
\SetKwProg{FnStageOne}{Stage 1}{}{}
\SetKwProg{FnStageTwo}{Stage 2}{}{}

\newcommand{\rev}[1]{#1}
\newcommand{\petar}[1]{#1}

\providecommand{\CfgBaseN}{500}
\providecommand{\CfgBaseM}{64}
\providecommand{\CfgBaseK}{16}
\providecommand{\CfgBaseL}{64}
\providecommand{\CfgBaseNs}{64}
\providecommand{\CfgBaseRank}{16}
\providecommand{\CfgAlphaCone}{0.7}
\providecommand{\CfgAlphaSUS}{0.1}

\providecommand{\CfgCandidateMult}{4}
\providecommand{\CfgCsiError}{0.1}
\providecommand{\CfgMinRsrp}{-120.0}
\providecommand{\CfgWmmseIters}{40}
\providecommand{\CfgWmmseTol}{0.0001}
\providecommand{\CfgTrials}{200}

\providecommand{\CfgMainN}{128}
\providecommand{\CfgMainL}{64}
\providecommand{\CfgMainNs}{64}
\providecommand{\CfgMainRank}{16}
\providecommand{\CfgMainAodError}{2.0}
\providecommand{\CfgMainPowerError}{1.0}
\providecommand{\CfgPfN}{500}
\providecommand{\CfgPfNs}{256}
\providecommand{\CfgPfSlots}{500}
\providecommand{\CfgPfTrajectories}{20}
\providecommand{\CfgPfTc}{50.0}
\providecommand{\CfgPfSnr}{20.0}
\providecommand{\CfgPfEpsilon}{0.001}

\providecommand{\CfgGpAnchorsPerSlot}{16}

\providecommand{\CfgGpWindowSlots}{20}

\providecommand{\CfgGpLengthScale}{0.5}
\providecommand{\CfgGpAodScale}{30.0}
\providecommand{\CfgGpNoiseLevel}{0.1}
\providecommand{\CfgGpSlots}{100}
\providecommand{\CfgGpTrajectories}{20}

\providecommand{\ProdSweepResultTag}{twc_v1}

\providecommand{\ProdMainTrials}{200}
\providecommand{\ProdMainPAtFifteen}{35.07}
\providecommand{\ProdMainPCIAtFifteen}{0.49}
\providecommand{\ProdMainLAtFifteen}{35.06}
\providecommand{\ProdMainLCIAtFifteen}{0.49}
\providecommand{\ProdMainSusAtFifteen}{30.28}
\providecommand{\ProdMainSusCIAtFifteen}{0.60}
\providecommand{\ProdMainAllSusAtFifteen}{41.79}
\providecommand{\ProdMainAllSusCIAtFifteen}{0.47}

\providecommand{\ProdMainLGainVsGrantedSusAtFifteen}{15.8}
\providecommand{\ProdMainLGainVsGrantedSusAtTwenty}{17.6}
\providecommand{\ProdPFTrajectories}{20}
\providecommand{\ProdPFSlots}{500}
\providecommand{\ProdPFSumRate}{17.68}
\providecommand{\ProdPFSumRateCI}{0.98}
\providecommand{\ProdPFJain}{0.180}
\providecommand{\ProdPFJainCI}{0.012}

\providecommand{\ProdGPTrajectories}{20}
\providecommand{\ProdGPSlots}{100}
\providecommand{\ProdGPNoRate}{42.99}
\providecommand{\ProdGPNoRateCI}{1.29}
\providecommand{\ProdGPCausalRate}{52.49}
\providecommand{\ProdGPCausalRateCI}{1.30}
\providecommand{\ProdGPOracleRate}{52.83}
\providecommand{\ProdGPOracleRateCI}{1.36}
\providecommand{\ProdLBase}{64}
\providecommand{\ProdLBaseP}{55.04}
\providecommand{\ProdLBasePCI}{0.57}
\providecommand{\ProdLBaseL}{54.76}
\providecommand{\ProdLBaseLCI}{0.56}
\providecommand{\ProdLBaseSus}{29.40}
\providecommand{\ProdLBaseSusCI}{0.71}
\providecommand{\ProdNsBase}{64}
\providecommand{\ProdNsBaseL}{54.74}
\providecommand{\ProdNsBaseLCI}{0.61}
\providecommand{\ProdNsMax}{500}
\providecommand{\ProdNsMaxL}{54.39}
\providecommand{\ProdNsMaxLCI}{0.63}
\providecommand{\ProdMBase}{64}
\providecommand{\ProdMBaseEta}{0.683}
\providecommand{\ProdMBaseEtaCI}{0.005}
\providecommand{\ProdMBaseRankPZeroFive}{29.0}
\providecommand{\ProdMBaseRankMedian}{32.0}
\providecommand{\ProdMBaseRankPNinetyFive}{36.0}
\providecommand{\ProdMBaseL}{54.81}
\providecommand{\ProdMBaseLCI}{0.54}
\providecommand{\ProdMBaseLAdaptive}{58.70}
\providecommand{\ProdMBaseLAdaptiveCI}{0.51}
\providecommand{\ProdKBase}{16}
\providecommand{\ProdKBaseL}{54.61}
\providecommand{\ProdKBaseLCI}{0.56}
\providecommand{\ProdAlphaDefault}{0.7}
\providecommand{\ProdAlphaDefaultL}{54.96}
\providecommand{\ProdAlphaDefaultLCI}{0.58}
\providecommand{\ProdAlphaDefaultCone}{27.8}
\providecommand{\ProdAlphaDefaultConeCI}{1.0}
\providecommand{\ProdAlphaMin}{0.5}
\providecommand{\ProdAlphaMinL}{56.94}
\providecommand{\ProdAlphaMinLCI}{0.49}
\providecommand{\ProdAlphaMax}{0.9}
\providecommand{\ProdAlphaMaxL}{53.55}
\providecommand{\ProdAlphaMaxLCI}{0.61}
\providecommand{\ProdDropMax}{0.7}
\providecommand{\ProdDropMaxL}{47.81}
\providecommand{\ProdDropMaxLCI}{0.56}
\providecommand{\ProdAodMax}{20}
\providecommand{\ProdAodMaxL}{36.76}
\providecommand{\ProdAodMaxLCI}{0.61}
\providecommand{\ProdPowerMax}{8}
\providecommand{\ProdPowerMaxL}{48.55}
\providecommand{\ProdPowerMaxLCI}{0.59}
\providecommand{\ProdPFDitusSumRate}{17.68}
\providecommand{\ProdPFDitusSumRateCI}{0.98}
\providecommand{\ProdPFDitusJain}{0.180}
\providecommand{\ProdPFDitusJainCI}{0.012}

\providecommand{\ProdPFDitusServiceCoverage}{0.670}
\providecommand{\ProdPFDitusServiceCoverageCI}{0.014}

\providecommand{\ProdPFMaxPowerSumRate}{6.69}
\providecommand{\ProdPFMaxPowerSumRateCI}{0.52}
\providecommand{\ProdPFMaxPowerJain}{0.185}
\providecommand{\ProdPFMaxPowerJainCI}{0.009}

\providecommand{\ProdPFMaxPowerServiceCoverage}{0.678}
\providecommand{\ProdPFMaxPowerServiceCoverageCI}{0.007}

\providecommand{\ProdPFDftScoreSumRate}{8.75}
\providecommand{\ProdPFDftScoreSumRateCI}{1.27}
\providecommand{\ProdPFDftScoreJain}{0.176}
\providecommand{\ProdPFDftScoreJainCI}{0.012}

\providecommand{\ProdPFDftScoreServiceCoverage}{0.585}
\providecommand{\ProdPFDftScoreServiceCoverageCI}{0.017}

\providecommand{\ProdPFRandomSumRate}{5.83}
\providecommand{\ProdPFRandomSumRateCI}{0.46}
\providecommand{\ProdPFRandomJain}{0.210}
\providecommand{\ProdPFRandomJainCI}{0.009}

\providecommand{\ProdPFRandomServiceCoverage}{0.813}
\providecommand{\ProdPFRandomServiceCoverageCI}{0.013}

\providecommand{\ProdPFMaxSrSumRate}{71.83}
\providecommand{\ProdPFMaxSrSumRateCI}{1.96}
\providecommand{\ProdPFMaxSrJain}{0.034}
\providecommand{\ProdPFMaxSrJainCI}{0.001}

\providecommand{\ProdPFMaxSrServiceCoverage}{0.089}
\providecommand{\ProdPFMaxSrServiceCoverageCI}{0.007}

\providecommand{\ProdPFRoundRobinSumRate}{12.05}
\providecommand{\ProdPFRoundRobinSumRateCI}{0.44}
\providecommand{\ProdPFRoundRobinJain}{0.202}
\providecommand{\ProdPFRoundRobinJainCI}{0.006}

\providecommand{\ProdPFRoundRobinServiceCoverage}{1.000}
\providecommand{\ProdPFRoundRobinServiceCoverageCI}{0.000}

\DeclareMathAlphabet{\mathcal}{OMS}{cmsy}{m}{n}

\begin{document}
\title{{Digital Twin-Aided Prescreening for User Scheduling in MU-MIMO Downlink Systems}}
\author{
Namhyun~Kim, Mahmoud~Saad~Abouamer, Jeonghun~Park, Petar~Popovski, and Ahmed~Alkhateeb%
\thanks{
N. Kim and A. Alkhateeb are with the School of Electrical, Computer and Energy Engineering, Arizona State University, Tempe, AZ 85287, USA (e-mail: namhyunk@asu.edu, alkhateeb@asu.edu). M. S. Abouamer and P. Popovski are with the Department of Electronic Systems, Aalborg University, 9220 Aalborg \O st, Denmark (e-mail: mahmoudabo@es.aau.dk, petarp@es.aau.dk). J. Park is with the School of Electrical and Electronic Engineering, Yonsei University, Seoul 03722, South Korea (e-mail: jhpark@yonsei.ac.kr).
}
}

\maketitle \setcounter{page}{1} 

\begin{abstract}
\rev{In dense deployments, massive multi-user multiple-input multiple-output (MU-MIMO) base stations can acquire instantaneous channel state information (CSI) for only a limited subset of users per scheduling interval, restricting multiuser diversity. We therefore propose Digital Twin User pre-Screening (DiTUS), a digital-twin (DT)-aided framework that identifies promising users before instantaneous CSI acquisition. DiTUS forms spatial covariances from DT-inferred departure angles and path powers. Optional Gaussian-process (GP) calibration mitigates path-power bias, while the dominant rank-$r$ eigenspace of the aggregate covariance yields common reference beams. It \emph{prescreens} the pool using DiTUS-P, a low-complexity projection-energy rule, or DiTUS-L, a greedy log-determinant rule that promotes spatial compatibility. A two-level protocol collects scalar beam reports from shortlisted users and requests $r$-dimensional effective-channel vectors only from the scheduled set. The framework also supports proportional-fair scheduling. At 15~dB under DT imperfections, simulations with 128 candidates, a 64-user effective-CSI acquisition budget, and a 64-user shortlist show that DiTUS-L achieves $\ProdMainLAtFifteen\!\pm\!\ProdMainLCIAtFifteen$~bps/Hz versus $\ProdMainSusAtFifteen\!\pm\!\ProdMainSusCIAtFifteen$~bps/Hz for semi-orthogonal user selection (SUS) with full-dimensional CSI from 64 users, demonstrating that DT-based prescreening preserves substantial multiuser-diversity gains by identifying strong, spatially compatible users before acquiring effective-channel vectors.}

\end{abstract}

\begin{IEEEkeywords}
Massive MIMO, user scheduling, digital twin, two-stage precoding, beam-based feedback, limited feedback, channel state information, proportional-fair scheduling.
\end{IEEEkeywords}

\section{Introduction}
\begingroup%

\rev{Massive multi-user multiple-input multiple-output (MU-MIMO) systems achieve high spectral efficiency by scheduling spatially compatible users and precoding their streams~\cite{Bjornson:2017,Goldsmith:2003,Tse:2005}. Each interval, the base station (BS) selects $K$ of $N$ candidates. With instantaneous channel state information (CSI) for a large pool, semi-orthogonal user selection (SUS) with zero-forcing (ZF) precoding can approach the dirty-paper-coding sum-rate benchmark~\cite{Yoo:2006,Lee:2017}. In dense deployments, however, full-pool CSI acquisition incurs substantial measurement and feedback overhead.}

\rev{Practical feedback permits instantaneous-CSI acquisition from at most $L$ users per interval. When $L\!\ll\!N$, SUS and related norm- or projection-based methods~\cite{Yoo:2006,Dimic:2005} search only the observed subset, limiting full-pool multiuser diversity. Two-stage and statistical precoding~\cite{Adhikary:2013,Lee:2017} reduce the instantaneous CSI dimension but require reliable long-term statistics or coarse feedback. The challenge is to identify promising users before requesting instantaneous CSI.}

\rev{Site-specific digital twins (DTs) provide angles of departure (AoDs) and average path powers from geometry-based propagation models~\cite{Wang:2023DTSurvey,DT:Alkhateeb:2023,position:Li:2025}. These slowly varying parameters~\cite{Molisch:2004} can be maintained without per-slot instantaneous CSI from every user, enabling DT-aided \emph{prescreening}: DT statistics reduce the pool to $N_s$ users before instantaneous reports determine the final schedule.}

\subsection{Related Work}

\subsubsection{User Scheduling in MU-MIMO}

\rev{MU-MIMO sum-rate scheduling is NP-hard~\cite{Dimic:2005}, motivating SUS~\cite{Yoo:2006}, norm- and greedy-selection methods~\cite{Dimic:2005,Bayesteh:2008}, proportional-fair scheduling~\cite{Viswanath:2002,Kushner:2004}, and deep reinforcement learning~\cite{Xie:2020RL}. Many formulations assume instantaneous CSI from all $N$ candidates although only $K\!\ll\!N$ are served; full-pool instantaneous-CSI acquisition can therefore dominate the overhead in dense deployments.}

\subsubsection{Two-Stage Precoding and Statistical Grouping}

\rev{Two-stage precoding~\cite{Adhikary:2013,Lee:2017} and hybrid beamforming~\cite{Alkhateeb:2014,Molisch:2017} use statistical outer and reduced-dimensional inner components~\cite{Peel:2005,Shi:2011,statistical:zhang:2017}. They reduce the instantaneous CSI dimension but require reliable long-term channel information, whose measurement can consume substantial pilot and feedback resources.}

\subsubsection{Digital Twins for Wireless Communications}

\rev{Digital twins and site-specific channel maps have been studied for wireless physical-layer design~\cite{Wang:2023DTSurvey,DT:Alkhateeb:2023}, including channel calibration and DT-aided subspace channel estimation~\cite{DT:Xu:2024,Alikhani2025DigitalTwin}, GP-assisted channel-statistics prediction from uncalibrated DTs and sparse measurements~\cite{Abouamer2025ChannelStatistics}, channel-map-based estimation~\cite{CKM:Jiang:2025}, pilot management~\cite{DT:Sarker:2026}, and Bayesian channel learning~\cite{position:Moro:2025}. Learned minimum mean-square error (MMSE) channel estimation provides a complementary data-driven approach~\cite{Neumann:2018}. Scheduling-level prescreening additionally requires identifying relevant multipath information, ranking candidates under imperfect DT predictions, calibrating systematic bias from available measurements, and enforcing an explicit feedback budget.}

\subsection{Contributions}
\rev{This paper proposes DiTUS (Digital Twin User pre-Screening), a DT-aided framework that ranks all $N$ candidates from DT statistics before per-slot reports. It decouples the scalar-report shortlist size $N_s$, selected according to scalar-feedback resources, from the effective-CSI acquisition budget $L$; only the final $K\le L$ users provide $r$-dimensional effective-channel vectors. The model separately accounts for scalar, effective-CSI, reference-signal, and calibration resources. DiTUS adopts Third Generation Partnership Project (3GPP) beam-based measurement and feedback concepts~\cite{3GPP:38214}.}

\rev{The principal contribution is an integrated prescreening framework that identifies AoDs and path powers as scheduling-relevant DT information, ranks candidates under imperfect predictions, optionally calibrates path-power bias, and embeds these components in a two-level protocol. DT covariances yield a dominant subspace and common reference beams. \petar{Each shortlisted user reports $(b_u,Q_u,c_u)$: a beam index, quantized channel quality indicator (CQI), and one-bit cone indicator. Only the final $K$ users report effective-channel vectors for weighted minimum mean-square error (WMMSE) precoding.} With $B_Q$ CQI bits and $B_g(r)$ bits per effective-channel vector, the uplink load is $N_s(\lceil\log_2r\rceil+B_Q+1)+K B_g(r)$ bits per instance; downlink reference and Gaussian-process (GP) anchor costs are reported separately. The variants are:}
\begin{itemize}
  \item \rev{\textbf{Projection-score prescreening (DiTUS-P):} A projection-energy score retains the top $N_s$ candidates. Under the predicted covariance model, Proposition~\ref{prop:prescreen-loss} upper-bounds the expected interference-free single-user matched-filter (MF) rate of every discarded user.}
  \item \rev{\textbf{Log-det prescreening (DiTUS-L):} Prescreening is formulated as monotone submodular log-determinant selection. Lemma~\ref{prop:logdet} gives the standard $(1-1/e)$ greedy guarantee for the aggregate-covariance objective over the DT-retained candidate set.}
\end{itemize}

\rev{We further extend DiTUS to proportional-fair (PF) scheduling by incorporating long-term throughput weights into prescreening, CQI ranking, and WMMSE precoding. Static weights can encode priorities but do not guarantee latency or minimum rates. Simulations evaluate the complete interference-coupled pipelines under DT imperfections and quantify the tradeoffs among shortlist size, fairness, calibration, and mismatch. Both DiTUS variants outperform $L$-user SUS throughout the main SNR sweep.}

\begin{figure}[!t]
  \centering
  \includegraphics[width=\columnwidth]{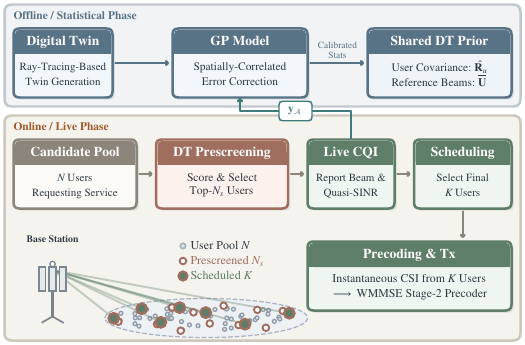}
  \caption{\rev{DiTUS architecture. DT statistics form $\hat{\mathbf R}_u$ and $\overline{\mathbf U}$ for prescreening and instantaneous CQI. An optional anchor measurement supplies a delayed observation for updating the GP in the next slot.}}
  \label{fig:system_model}
\end{figure}

\textit{Notation}:
Throughout this paper, bold lowercase letters ($\mathbf{h}$) denote vectors, bold uppercase letters ($\mathbf{H}$) denote matrices, and calligraphic letters ($\mathcal{D}$) denote sets or datasets. $\mathbb{E}[\cdot]$ denotes expectation, $\mathbb{C}$ is the set of complex numbers, and $\mathbb{R}$ is the set of real numbers. $\|\cdot\|$ denotes the Euclidean norm, $\operatorname{diag}(\cdot)$ forms a diagonal matrix, and $\mathcal{CN}(0,\sigma^2)$ denotes a circularly symmetric complex Gaussian distribution with zero mean and variance $\sigma^2$. The superscript $\sf{H}$ denotes Hermitian transpose, and $\sf{T}$ denotes transpose. Other notation is defined as needed in the text.

\endgroup
\section{System Model}

\subsection{Received Signal Model}
We consider a single BS equipped with an $M$-antenna uniform linear array (ULA). In each slot, the BS selects $K$ single-antenna users from a pool of $N\gg K$ requesting users and serves them over a block-fading channel~\cite{Tse:2005}.\footnote{\rev{We use a ULA for analytical clarity; the covariance construction and prescreening scores extend to other arrays by substituting the appropriate steering vector.}} The received signal vector $\mathbf{y}\in\mathbb{C}^{K\times 1}$ is
\begin{equation}\label{rsModel:2}
    \mathbf{y} = \mathbf{H}\mathbf{F}\mathbf{s} + \mathbf{n},
\end{equation}
where $\mathbf{H} = [\mathbf{h}_1,\ldots,\mathbf{h}_K]^{\sf H}\in\mathbb{C}^{K\times M}$ is the channel matrix stacking individual user channel vectors row-wise, $\mathbf{F}\in\mathbb{C}^{M\times K}$ is the precoder, \rev{$\mathbf{s}\in\mathbb{C}^{K\times 1}$ contains unit-power transmitted symbols with $\mathbb{E}[|s_{k}|^2]=1$, and $\mathbf{n}\sim\mathcal{CN}(0,\sigma^2\mathbf{I}_K)$ is additive white Gaussian noise (AWGN).}
\rev{Here, $k\in\{1,\ldots,K\}$ indexes a transmitted stream and its scheduled user, whereas $u\in\{1,\ldots,N\}$ denotes a user identity in the full candidate pool.}
For each scheduled user $k$, we define the long-term channel second moment (spatial covariance) as
\[
\mathbf{R}_{k}\triangleq \mathbb{E}\!\left[\mathbf{h}_{k}\mathbf{h}_{k}^{\sf H}\right].
\]
Per-user DT-predicted covariances $\hat{\mathbf{R}}_u$, constructed from long-term DT path parameters without instantaneous CSI, are detailed in Section~II-D.
For user $k$, $y_{k}=\mathbf{h}_{k}^{\sf H}\mathbf{f}_{k}s_{k}+\sum_{j\neq k}\mathbf{h}_{k}^{\sf H}\mathbf{f}_{j}s_{j}+n_{k}$, yielding the signal-to-interference-plus-noise ratio (SINR)
\begin{equation}
    \gamma_{k} = \frac{|\mathbf{h}_{k}^{\sf H}\mathbf{f}_{k}|^2}{\sum_{j\neq {k}}|\mathbf{h}_{k}^{\sf H}\mathbf{f}_{j}|^2+\sigma^2}.
\end{equation}
The corresponding sum rate defines the objective in the following problem formulation.

\subsection{Problem Formulation}
\rev{We consider joint user-set and precoder design for sum-rate maximization with $N \gg K$ candidates. The BS serves $K$ users per slot, reflecting the available spatial degrees of freedom, feedback resources, and channel conditioning. Let $\mathcal K\subseteq\{1,\ldots,N\}$ be the scheduled set and $\mathbf F_{\mathcal K}=[\mathbf f_u]_{u\in\mathcal K}\in\mathbb C^{M\times K}$ its precoder.}

\begin{subequations} \label{eq:optimization}
\begin{align}
\underset{\mathcal K,\ \mathbf{F}_{\mathcal K}}{\text{maximize}} \quad
& \sum_{u\in\mathcal K} \log_2(1 + \gamma_{u}) \label{eq:optimization-a} \\
\text{subject to} \quad
& |\mathcal K| = K, \quad \sum_{u\in\mathcal K} \|\mathbf{f}_{u}\|^2 \leq P. \label{eq:optimization-b}
\end{align}
\end{subequations}
\rev{The joint problem is combinatorial because each candidate set induces a different interference-coupled channel matrix and hence a different optimal precoder~\cite{Dimic:2005,Bjornson:2014}. DiTUS addresses this by DT-based shortlisting ahead of any per-slot report, as detailed in Sec.~IV.}

\subsection{Duplexing and Online Resource Model}\label{sec:resource-model}
\rev{We adopt a frequency-division duplex (FDD) acquisition model based on New Radio (NR) CSI reference-signal (CSI-RS) measurement and feedback~\cite{3GPP:38211,3GPP:38214,Giordani:2019}. Candidate identities and locations are maintained on a slower control-plane timescale. Per scheduling instance, the BS broadcasts a common $r$-port beamformed reference signal with time--frequency cost $C_{\rm RS}(r)$ independent of the number of measuring users. Each $u\in\mathcal S$ returns $(b_u,Q_u,c_u)$ using $\lceil\log_2r\rceil+B_Q+1$ bits; after $\mathcal K$ is fixed, its $K$ users each return $B_g(r)$ bits for a quantized $r$-dimensional effective-channel vector. Thus, $K\le L$, $K\le N_s\le N$, and $K\le r\le M$, where $L$ is the per-slot effective-CSI acquisition budget, not a reference-signal-symbol count. Because first-stage reports are scalar, $N_s$ may exceed $L$ when scalar-feedback resources permit, but such points incur additional feedback and are not equal-total-overhead comparisons without a common bit mapping.}

The GP path-power label is a distinct slow-timescale resource: recovering a path-resolved power error requires calibrated sounding or sensing plus DT-to-path association and cannot be obtained from the standard effective-CSI report. We denote its measurement and feedback costs by $C_{\rm A}$ and $B_{\rm A}$ and report them separately.
DiTUS-L performs a single prescreening pass and does not refine the shortlist through additional CSI-feedback stages.

\subsection{Digital Twin Channel Modeling and Covariance Estimation}
For the scheduling problem considered here, the DT provides BS-side, geometry-aware priors on long-term channel statistics, rather than instantaneous CSI. Specifically, each candidate user $u$ is associated with a DT-predicted path set
\begin{equation}\label{eq:dt-pathset}
    \mathcal{P}_u^{\mathrm{DT}} = \{(\hat{\theta}_{u,\ell},\hat{p}_{u,\ell})\}_{\ell=1}^{\hat{L}_u},
\end{equation}
where $\hat{L}_u$ is the number of DT-predicted paths for user $u$ (the DT-side counterpart of the true path count $L_u$ in the geometric channel model~\eqref{eq:geo-channel}), and $\hat{\theta}_{u,\ell}$ and $\hat{p}_{u,\ell}$ denote the predicted azimuth AoD and average path power of the $\ell$-th path, respectively.

\subsubsection{Geometric Channel Model}
\rev{The physical channel vector $\mathbf{h}_u \in \mathbb{C}^{M \times 1}$ for candidate user $u \in \{1,\ldots,N\}$ is modeled as a superposition of $L_u$ propagation paths~\cite{Saleh:1987,3GPP:38901}:}
\begin{equation}\label{eq:geo-channel}
    \mathbf{h}_u = \sum_{\ell=1}^{L_u} \sqrt{p_{u,\ell}} e^{-j \phi_{u,\ell}} \mathbf{a}(\theta_{u,\ell}),
\end{equation}
where $p_{u,\ell}$, $\phi_{u,\ell}$, and $\theta_{u,\ell}$ denote the path power, phase, and azimuth AoD of the $\ell$-th path, respectively. In a wideband model, the path delay would induce a frequency-dependent phase rotation; at the single narrowband frequency considered here, this rotation is absorbed into $\phi_{u,\ell}$ and does not affect the spatial covariance. The array steering vector of the ULA with inter-element spacing \rev{$d=\lambda/2$}~\cite{Molisch:2004} is
\begin{equation}\label{eq:array-response}
\mathbf{a}(\theta)\triangleq \frac{1}{\sqrt{M}}
\big[e^{j2\pi (d/\lambda)m\sin\theta}\big]_{m=0}^{M-1}
\in\mathbb{C}^{M\times 1}.
\end{equation} \rev{Under the convention in~\eqref{rsModel:2}, a user's channel enters the received-signal model as $\mathbf h_u^{\sf H}$; hence both the DiTUS-P projection score in~\eqref{eq:prescreen-score} and construction of the effective-channel vector use the projection $\mathbf a(\theta)^{\sf H}\overline{\mathbf U}$.} While the phases $\phi_{u,\ell}$ vary rapidly, the structural parameters $\{\theta_{u,\ell}, p_{u,\ell}\}$ remain relatively stable over a longer coherence interval dependent on user mobility (typically hundreds of milliseconds to seconds)~\cite{Molisch:2004}.

\subsubsection{Digital Twin Imperfections}
The uncalibrated DT predictions are subject to estimation errors due to sensing inaccuracies and synchronization mismatch. We model these imperfections as
\begin{align}
    \hat{\theta}_{u,\ell} &= \theta_{u,\ell} + n_{u,\ell}^{\theta}, \label{eq:angle-noise} \\
    10\log_{10}(\hat{p}_{u,\ell}) &= 10\log_{10}(p_{u,\ell}) + n_{u,\ell}^{P}, \label{eq:power-noise}
\end{align}
where $n_{u,\ell}^{\theta}$ and $n_{u,\ell}^{P}$ denote the AoD and power errors, respectively. \rev{The additive model specifies how DT parameter mismatch enters the covariance construction while leaving the error distributions general.}

\subsubsection{AoD-Dependent Path-Power Error Correction via Gaussian Processes}
DT predictions can exhibit systematic path-power bias, for example because of imperfect ray-tracing calibration. We represent the predicted AoD by the circular feature $\mathbf{x}(\hat\theta)=[\cos\hat\theta,\sin\hat\theta]^{\sf T}/\theta_0$, where $\theta_0=\CfgGpAodScale^\circ=\pi/6$, and model the logarithmic power error $e(\mathbf{x})\triangleq10\log_{10}(p)-10\log_{10}(\hat p)$ as a GP:
\begin{equation}
    e(\mathbf{x}) \sim \mathcal{GP}\big(m(\mathbf{x}), k(\mathbf{x}, \mathbf{x}')\big),
\end{equation}
with mean function $m(\mathbf{x})=0$ and radial basis function (RBF) kernel $k(\mathbf{x},\mathbf{x}')=\sigma_f^2\exp(-\|\mathbf{x}-\mathbf{x}'\|^2/(2\ell^2))$. We use $\ell=\CfgGpLengthScale$ and unit kernel variance, $\sigma_f^2=1$, after normalizing the observed error labels by their sample standard deviation. The circular feature makes $-180^\circ$ and $180^\circ$ adjacent and keeps the kernel scale fixed across slots.

At the beginning of slot $t$, the GP may use only path-resolved labels acquired no later than slot $t-1$, denoted by $\mathcal A^{(t-1)}$. Each label $y_i$ is a calibrated path-power error acquired through the anchor resource of Section~\ref{sec:resource-model}. Let $s_{\mathcal A}$ be the sample standard deviation of the available labels, with $s_{\mathcal A}=1$ for a constant or empty set, and let $\tilde{\mathbf y}_{\mathcal A}=\mathbf y_{\mathcal A}/s_{\mathcal A}$. The normalized posterior mean for a candidate path is
\begin{equation}
    \tilde\mu_{\mathcal{A}^{(t-1)}}(\mathbf{x}^*) = \mathbf{k}_*^{\sf T} (\mathbf{K} + \sigma_n^2 \mathbf{I})^{-1} \tilde{\mathbf y}_{\mathcal{A}^{(t-1)}},
\end{equation}
where $\mathbf{K}$ is the Gram matrix of anchor AoD features, $\mathbf{k}_*$ is the covariance vector between the target feature $\mathbf{x}^*$ and the anchor features, and $\sigma_n^2=\CfgGpNoiseLevel$ is a dimensionless diagonal regularization parameter. Rescaling gives $\mu_{\mathcal{A}}=s_{\mathcal A}\tilde\mu_{\mathcal A}$ in decibels, and the DT path powers are corrected as $\hat{p}_{u,\ell}^{\mathrm{corr}}=\hat{p}_{u,\ell}\,10^{\mu_{\mathcal{A}^{(t-1)}}(\mathbf{x}(\hat{\theta}_{u,\ell}))/10}$.

\rev{To ensure causality, we initialize $\mathcal A^{(0)}=\emptyset$ and add labels collected after slot $t$ only to $\mathcal A^{(t)}$, so they first affect slot $t+1$. After each slot, a fixed number of scheduled users may be measured, and a finite sliding window limits the retained labels. Missing or unmatched paths are excluded. The evaluated GP assumes stable path association and is therefore not combined with path dropout. Away from anchors, the zero-mean GP reverts to the uncorrected DT prediction. The chronology is $\mathcal A^{(t-1)}\!\rightarrow$ GP $\rightarrow\mathcal S^{(t)}\!\rightarrow\mathcal K^{(t)}\!\rightarrow\mathcal A^{(t)}$.}

\subsubsection{Covariance Construction}
The corrected long-term covariance of user $u$ is constructed from the calibrated DT path parameters. Under the widely adopted uncorrelated scattering assumption~\cite{Bello:1963,Molisch:2004} (i.e., independent uniformly distributed phases for distinct propagation paths), the cross-terms vanish in the expectation. Thus, the spatial covariance is
\begin{equation}\label{eq:corrected-cov}
    \hat{\mathbf{R}}_u = \sum_{\ell=1}^{\hat{L}_u} \hat{p}_{u,\ell}^{\mathrm{corr}} \mathbf{a}(\hat{\theta}_{u,\ell}) \mathbf{a}(\hat{\theta}_{u,\ell})^\mathsf{H}.
\end{equation}

\section{Proposed Two-Stage Precoding Framework}
Building on~\cite{Lee:2017}, the outer beamformer is constructed from DT-inferred covariances rather than reference-signal-based covariance estimates, and DT-based shortlisting precedes CQI reporting. The structure uses NR-style beam-based training and feedback~\cite{Giordani:2019,3GPP:38214}.

The BS uses a slowly varying \emph{dominant spatial subspace} to define the reference beams. Let $\mathbf{R}=N^{-1}\sum_{u=1}^N\mathbf R_u$ denote the pool-average long-term channel covariance. DiTUS estimates this quantity from DT-derived per-user covariances without acquiring instantaneous CSI from all candidates. The eigendecomposition of $\mathbf{R}$ is
\begin{equation}
    \mathbf{R} = \mathbf{U} \mathbf{\Sigma} \mathbf{U}^{\sf H},
\end{equation}
where $\mathbf{U} \in \mathbb{C}^{M \times M}$ contains the eigenvectors and $\mathbf{\Sigma}=\operatorname{diag}(\lambda_1,\ldots,\lambda_M)$ contains eigenvalues sorted in descending order.
To capture the dominant spatial directions, we define a reduced eigenvector matrix (outer beamformer)
\begin{equation}
    \overline{\mathbf{U}} \triangleq [\mathbf{u}_1,\ldots,\mathbf{u}_r]\in\mathbb{C}^{M\times r},\quad \rev{K\le r\le M},
\end{equation}
which serves as the reference-beam matrix.

For candidate user $u$, the \emph{instantaneous effective CSI vector} is
\begin{equation}\label{eq:effective-csi-vector}
    \tilde{\mathbf h}_u \triangleq \overline{\mathbf U}^{\sf H}\mathbf h_u \in \mathbb C^r.
\end{equation}

To identify users whose channels are aligned with the dominant eigenspace, we define a \textit{selection cone} for each reference beam $\mathbf{u}_i$ as
\begin{equation}\label{eq:cone-def}
    \mathcal{C}_{i} = \left\{ \mathbf{h}: i = \arg\max_{m}|\mathbf{h}^{\sf H} \mathbf{u}_{m}|^2,\ \tfrac{|\mathbf{h}^{\sf H} \mathbf{u}_{i}|^2}{\|\mathbf{h}^{\sf H}\overline{\mathbf{U}}\|^2} \geq \alpha_{\rm cone} \right\},
\end{equation}
for $i=1,\ldots,r$, where $\alpha_{\rm cone}\in[1/2,1)$ is the cone threshold. Ties in $\arg\max$ are resolved by the smallest beam index, including the $\alpha_{\rm cone}=1/2$ boundary, so every user is assigned to at most one cone.

\begin{lemma}[Cone condition and pairwise semi-orthogonality~\cite{Lee:2017}]\label{prop:cone-zf}
Let $\alpha_{\rm cone}\in[1/2,1)$. For any effective dimension $r\ge 2$, if users $k,\ell$ are assigned to distinct cones $i\ne j$, their normalized effective-channel vectors $\tilde{\mathbf{h}}_k$ and $\tilde{\mathbf{h}}_\ell$ satisfy
\begin{equation}\label{eq:cone-inner}
\frac{|\tilde{\mathbf{h}}_k^{\sf H}\tilde{\mathbf{h}}_\ell|^2}{\|\tilde{\mathbf{h}}_k\|^2\|\tilde{\mathbf{h}}_\ell\|^2}
\le 4\alpha_{\rm cone}(1-\alpha_{\rm cone}),
\end{equation}
and the bound is attainable by channels supported on the two corresponding cone coordinates.
\end{lemma}
\begin{IEEEproof}
Applying the Cauchy--Schwarz argument of Lemma~2 in~\cite{Lee:2017} yields~\eqref{eq:cone-inner}.
\end{IEEEproof}
\rev{By Lemma~\ref{prop:cone-zf}, assigning users to distinct cones controls pairwise effective-channel correlation.} The bound $4\alpha_{\rm cone}(1-\alpha_{\rm cone})$ decreases on $[1/2,1]$, so a larger threshold tightens the pairwise-correlation bound at the cost of fewer cone-qualified users.

Let $\hat{\mathbf h}_u$ denote the channel estimate obtained from the common reference signal. A shortlisted user locally forms $\hat{\mathbf g}_u\triangleq\hat{\mathbf h}_u^{\sf H}\overline{\mathbf U}$ and reports its most-aligned beam
\begin{equation}
    b_u = \arg\max_{i\in\{1,\ldots,r\}} |\hat g_{u,i}|^2,
\end{equation}
along with a quasi-SINR metric based on its estimated effective-channel energy and a one-bit cone-admissibility indicator $c_u$
\begin{equation}\label{eq:qk-single}
    Q_u := \frac{\|\hat{\mathbf g}_u\|^2}{\Delta},\quad \Delta\triangleq \frac{K\sigma^2}{P},
\end{equation}
where $\Delta$ is the normalized noise term under equal power allocation across $K$ streams. \rev{Equivalently, $Q_u=\rho_{\mathrm{tx}}\|\hat{\mathbf g}_u\|^2$, where $\rho_{\mathrm{tx}}=P/(K\sigma^2)$, so $Q_u$ is proportional to the per-stream post-beamforming SNR.}
Here $c_u=\mathbf 1\{|\hat g_{u,b_u}|^2/\|\hat{\mathbf g}_u\|^2\ge\alpha_{\rm cone}\}$.

At the BS, cone-qualified users are first pooled by reported beam index $b_u$; for each beam $i$ the BS selects
\begin{equation}
    \kappa_i = \arg\max_{u\in\mathcal{W}_i} Q_u,\qquad i\in\mathcal I_{\rm occ},
\end{equation}
where $\mathcal{W}_i$ is the set of users with $b_u=i$ and $c_u=1$, and $\mathcal I_{\rm occ}\triangleq\{i:\mathcal W_i\neq\emptyset\}$ is the set of populated cones. The BS retains at most the top-$K$ users among $\{\kappa_i\}_{i\in\mathcal I_{\rm occ}}$ based on $Q_{\kappa_i}$. If fewer than $K$ users are obtained after per-beam competition, the remaining positions are filled from the other shortlisted users using descending reported $Q$; no candidate's effective-channel vector is available to or used by the scheduler.

In the second feedback stage, each scheduled user reports its instantaneous effective CSI, from which the BS computes the weighted minimum mean-square error (WMMSE) inner precoder for interference mitigation~\cite{Shi:2011}. Thus, instantaneous CQI refines the DT-based shortlist before effective-channel vectors are acquired.

\subsection{Second Stage: Precoding}\label{sec:second-stage}
We employ two-stage precoding to mitigate inter-user interference~\cite{Lee:2017}. \rev{The reference-beam matrix $\overline{\mathbf{U}}\in\mathbb{C}^{M\times r}$, $r\ge K$, serves as the outer precoder, and the inner precoder $\mathbf{V}\in\mathbb{C}^{r\times K}$ is designed from instantaneous effective CSI, yielding the full precoder $\mathbf{F}_{\mathrm{ts}}=\overline{\mathbf{U}}\mathbf{V}$.} With $\tilde{\mathbf{H}}\triangleq\mathbf{H}\overline{\mathbf{U}}\in\mathbb{C}^{K\times r}$, \eqref{rsModel:2} becomes $\mathbf{y}=\tilde{\mathbf{H}}\mathbf{V}\mathbf{s}+\mathbf{n}$. The inner precoder maximizes the sum spectral efficiency subject to the total power constraint:
\begin{equation}\label{eq:wmmse-problem}
\rev{\max_{\mathbf{V}:\ \|\overline{\mathbf{U}}\mathbf{V}\|_F^2\le P} \sum_{k=1}^K \log_2(1 + \gamma_k),}
\end{equation}
where $\gamma_k$ is the SINR of user $k$ with effective-channel vector $\tilde{\mathbf{h}}_k$. A stationary point of this nonconvex problem is obtained via the iterative WMMSE algorithm~\cite{Shi:2011}, which alternately updates the scalar receive equalizer $a_k$, MSE weight $\nu_k=e_k^{-1}$, and the precoder:
\begin{equation}\label{eq:wmmse-update}
\mathbf{V} = \left(\sum_{k=1}^K \nu_k |a_k|^2 \tilde{\mathbf{h}}_k \tilde{\mathbf{h}}_k^{\sf H} + \lambda \mathbf{I}_r\right)^{-1} \left[ \nu_1 a_1^* \tilde{\mathbf{h}}_1, \dots, \nu_K a_K^* \tilde{\mathbf{h}}_K \right],
\end{equation}
where $\lambda$ is chosen to satisfy the power constraint. Since $\overline{\mathbf U}^{\sf H}\overline{\mathbf U}=\mathbf I_r$, $\|\overline{\mathbf U}\mathbf V\|_F=\|\mathbf V\|_F$; for weighted sum-rate maximization, set $\nu_k=w_k(t)e_k^{-1}$, where $w_k(t)$ is the PF weight defined in Section~V.
\rev{The following proposition upper-bounds the loss incurred when an interference-free, equal-power surrogate is restricted to the $r$-dimensional subspace. The optimized WMMSE performance is assessed numerically in Section~VI.}

\begin{proposition}[Interference-free subspace-loss surrogate]\label{prop:rate-gap}
\rev{Let $R_{\rm full}^{\rm IF}=\sum_k\log_2(1+\rho_{\mathrm{tx}}\|\mathbf h_k\|^2)$ and $R_{\rm sub}^{\rm IF}=\sum_k\log_2(1+\rho_{\mathrm{tx}}\|\overline{\mathbf U}^{\sf H}\mathbf h_k\|^2)$ be equal-power, interference-free surrogate rates, where $\rho_{\mathrm{tx}}=P/(K\sigma^2)$. With $\epsilon_k\triangleq \|(\mathbf{I}_M-\overline{\mathbf{U}}\overline{\mathbf{U}}^{\sf H})\mathbf{h}_k\|^2$,}
\begin{equation}\label{eq:rate-gap-bound}
R_{\rm full}^{\rm IF}-R_{\rm sub}^{\rm IF} \le \sum_{k=1}^{K}\log_2\!\Big(1+\rho_{\mathrm{tx}}\,\epsilon_k\Big).
\end{equation}
\end{proposition}
\begin{IEEEproof}
See Appendix~\ref{app:prop2}.
\end{IEEEproof}
\rev{Together, Lemma~\ref{prop:cone-zf} and Proposition~\ref{prop:rate-gap} provide design guidance for cone-based user separation and subspace-rank selection, respectively. The WMMSE inner precoder mitigates residual multiuser interference, and the resulting interference-coupled performance is evaluated numerically.}

\begin{remark}[Rank selection guideline]\label{rem:rank}
\rev{Under the covariance-matching condition $\mathbf R_{\mathcal K}=\mathbf R$ used in Appendix~\ref{app:prop2}, the expected surrogate bound is controlled by the residual eigenvalue mass outside $\operatorname{span}(\overline{\mathbf U})$. Define $\eta(r)\triangleq\sum_{i=1}^{r}\lambda_i/\sum_i\lambda_i$, $r_\eta\triangleq\min\{r:\eta(r)\ge\eta_0\}$, and the adaptive rank $r_{\rm adapt}\triangleq\min\{r_{\max},\max(K,r_\eta)\}$. The antenna-scaling experiment in Section~VI applies this rule to the pool-average predicted DT covariance with $\eta_0=0.9$ and $r_{\max}=M$. This heuristic controls captured covariance energy rather than the optimized WMMSE rate gap.}
\end{remark}

\section{Proposed User Scheduling with Digital Twin}
We propose a DT-aided framework for scheduling dense MU-MIMO systems with $N\gg K$. DT-inferred long-term channel statistics define a scheduling-relevant spatial subspace and prescreen the candidate pool before instantaneous feedback is requested. Consequently, only $N_s$ candidates return scalar CQI, and effective-channel vectors are acquired only from the final $K$ scheduled users rather than from all $N$ candidates. \rev{Table~\ref{tab:notation} summarizes the key symbols.}

\begin{table}[!t]
\centering
\caption{\rev{Summary of key notation used in Section~IV.}}
\label{tab:notation}
\rev{\scriptsize
\setlength{\tabcolsep}{3pt}
\renewcommand{\arraystretch}{0.95}
\begin{tabular}{@{}ll@{}}
\toprule
Symbol & Meaning \\
\midrule
$N,K,N_s,M,r$ & pool, scheduled, shortlist, antennas, rank \\
$L$ & effective-CSI acquisition budget \\
$\overline{\mathbf{U}}$ & reference beamformer (eigen-beams of $\hat{\mathbf{R}}$) \\
$\hat{\mathbf{R}}_u,\hat{\mathbf{R}}$ & user-$u$ and pool-average DT covariances \\
$\hat{\mathbf{h}}_u,\hat{\mathbf{g}}_u$ & estimated channel/effective-channel vectors \\
$\hat\theta_{u,\ell},\hat p_{u,\ell}^{\mathrm{corr}}$ & DT-predicted AoD and GP-corrected path power \\
$b_u,Q_u,c_u$ & beam index, quasi-SINR CQI, cone indicator \\
$\alpha_{\rm cone},\alpha_{\rm SUS},N_{\rm cand}$ & cone threshold, SUS threshold, log-det candidate count \\
$\mathcal{S}^{(t)},\mathcal{K}^{(t)}$ & shortlist and scheduled set in slot $t$ \\
\bottomrule
\end{tabular}}
\end{table}

\subsection{Covariance Construction from DT Long-Term Parameters}
DT assistance is used to construct the dominant subspace $\overline{\mathbf{U}}$ and to score users without relying on instantaneous CSI.
From the DT path set $\mathcal{P}_u^{\mathrm{DT}}$ in~\eqref{eq:dt-pathset}, the covariance construction uses the long-term parameters $\{\hat\theta_{u,\ell},\hat p_{u,\ell}^{\mathrm{corr}}\}_{\ell=1}^{\hat{L}_u}$, where $\hat p_{u,\ell}^{\mathrm{corr}}$ denotes the GP-corrected average path power.
We use the per-user statistical covariance $\hat{\mathbf{R}}_u$ defined in~\eqref{eq:corrected-cov} (with summation up to $\hat{L}_u$), which follows from the narrowband geometric channel model in~\eqref{eq:geo-channel} under the same uncorrelated-scattering assumption.
We also form an aggregated cell covariance
\begin{equation}\label{eq:cell-cov}
\hat{\mathbf{R}}\triangleq \frac{1}{N}\sum_{u=1}^{N}\hat{\mathbf{R}}_u.
\end{equation}
Here, $N$ denotes the candidate-pool size within the scheduling instance. The factor $1/N$ affects only the overall scale and does not change the eigenvectors, hence it does not affect the reference-beam directions. We use the pool-average $\hat{\mathbf{R}}$ to obtain a \emph{single} set of reference beams shared across candidates, which enables a common reference-signal and feedback structure in the instantaneous-CQI stage. DT paths without valid parameters are omitted from the covariance sum.
These statistics require no per-candidate instantaneous CSI. The BS updates the user covariances, pool covariance, and dominant eigenspace on the slower control-plane timescale when the candidate set, DT paths, or GP anchors change, and reuses them between updates. Shortlist CQI, scheduled-user effective CSI, and the inner WMMSE precoder are updated every slot.
The reference-beam matrix $\overline{\mathbf{U}}$ is chosen as the $r$ dominant eigenvectors of $\hat{\mathbf{R}}$ (equivalently, the dominant eigen-beams of an AoD--power weighted covariance). The pair $(\hat{\mathbf{R}}_u,\overline{\mathbf{U}})$ serves as the input to the prescreening score introduced in Section~\ref{sec:dt-prescreen}.

\subsection{DiTUS-P: Projection-Score Prescreening and Two-Stage Feedback}\label{sec:dt-prescreen}
We first introduce \textbf{DiTUS-P}, the low-complexity projection-score variant of DiTUS. DiTUS-P uses DT-inferred statistical CSI for outer precoding and prescreening, and then uses lightweight CQI to refine selection. Instantaneous \emph{effective} CSI is requested only from the final \(K\) users.
\begin{definition}[DiTUS-P projection-score prescreening]\label{def:dt-prescreen}
For each candidate user $u\in\{1,\ldots,N\}$, the DT outputs the long-term path set $\mathcal{P}_u^{\mathrm{DT}}$ in~\eqref{eq:dt-pathset}. The BS uses only the AoDs and GP-corrected average path powers from this set to form a predicted long-term user covariance $\hat{\mathbf{R}}_u\in\mathbb{C}^{M\times M}$. With the common reference-beam matrix $\overline{\mathbf{U}}$ formed as the $r$ dominant eigenvectors of the pool-average covariance $\hat{\mathbf{R}}$, the BS assigns the prescreening score
\begin{equation}\label{eq:prescreen-score}
    s_u \triangleq \mathrm{tr}\!\left(\overline{\mathbf{U}}^{\sf H}\hat{\mathbf{R}}_u\,\overline{\mathbf{U}}\right)
    = \sum_{\ell}\hat p_{u,\ell}^{\mathrm{corr}}\|\mathbf{a}(\hat\theta_{u,\ell})^{\sf H}\overline{\mathbf{U}}\|^2,
\end{equation}
    which is the expected projection energy of user $u$ onto the dominant subspace under the predicted covariance model, and selects the candidate subset
\begin{equation}\label{eq:prescreen-set}
\mathcal{S}\triangleq \arg\max_{\mathcal{T}\subseteq\{1,\ldots,N\}:|\mathcal{T}|=N_s}\ \sum_{u\in\mathcal{T}} s_u,
\end{equation}
which is equivalently obtained by choosing the top-$N_s$ users according to $s_u$.
\end{definition}

Under the predicted covariance model, the score~\eqref{eq:prescreen-score} equals the expected effective-channel energy under the outer beamformer, $s_u = \mathbb{E}[\|\mathbf{h}_u^{\sf H}\overline{\mathbf{U}}\|^2]$. Under equal power allocation, it is proportional to an expected per-stream SNR contribution. Top-$N_s$ ranking therefore retains users with the largest individual energy proxies, but it does not by itself optimize multi-user spatial compatibility.

The following proposition bounds the expected single-user MF rate of discarded users.

\begin{proposition}[Prescreening-score bound on the expected single-user matched-filter rate]\label{prop:prescreen-loss}
\rev{Under the covariance model used for prescreening, $\mathbb{E}[\mathbf{h}_u\mathbf{h}_u^{\sf H}]=\hat{\mathbf{R}}_u$. Let $R_u^{\mathrm{MF}}\triangleq \log_2(1+\rho_{\mathrm{tx}}\|\mathbf{h}_u^{\sf H}\overline{\mathbf{U}}\|^2)$ denote the single-user matched-filter (MF) rate after outer beamforming, and let $s_{(N_s)}$ denote the $N_s$-th largest prescreening score among the $N$ candidate users. For any user $u$ not included in the prescreened set $\mathcal{S}$ (i.e., $s_u\le s_{(N_s)}$), the expected single-user MF rate satisfies}
\begin{equation}\label{eq:prescreen-bound}
\mathbb{E}\!\big[R_u^{\mathrm{MF}}\big]
\le \log_2\!\big(1+\rho_{\mathrm{tx}}\,s_{(N_s)}\big).
\end{equation}
Consequently, the sum of the expected individual MF rates of the discarded users is at most
\begin{equation}\label{eq:total-loss}
\sum_{u\notin\mathcal{S}} \mathbb{E}[R_u^{\mathrm{MF}}]
\le (N-N_s)\,\log_2\!\big(1+\rho_{\mathrm{tx}}\,s_{(N_s)}\big).
\end{equation}
\end{proposition}
\begin{IEEEproof}
See Appendix~\ref{app:thm1}.
\end{IEEEproof}

\begin{remark}[Practical implication]
The bound in~\eqref{eq:total-loss} is most informative when $(N-N_s)\log_2(1+\rho_{\mathrm{tx}}s_{(N_s)})$ is small; even small $s_{(N_s)}$ may give a loose bound when $N-N_s$ is large.
\end{remark}

\begin{remark}[Scope of the bound]
\rev{Proposition~\ref{prop:prescreen-loss} upper-bounds the expected single-user MF rate of discarded users under the assumed covariance model. Cone-based CQI selection and, for DiTUS-L, aggregate-covariance selection then promote spatial compatibility before the WMMSE inner precoder mitigates residual multiuser interference.}
\end{remark}

\rev{After prescreening, only users in $\mathcal{S}$ estimate their effective-channel vectors and report the scalar tuple $(b_u,Q_u,c_u)$ specified in~\eqref{eq:qk-single} and the surrounding text; the vectors themselves are requested only from the final $K$ users. Algorithm~\ref{alg:ditus} summarizes DiTUS-P. DiTUS-L replaces its projection-score prescreening step with the greedy objective in~\eqref{eq:new-logdet-objective}. The optional GP-anchor update applies to either variant; the causal-GP evaluation in Section~\ref{sec:causal-gp} uses DiTUS-L.}

\begin{algorithm}[t]
\caption{DiTUS-P User Scheduling \& Precoding}
\label{alg:ditus}
\SetAlgoLined
\KwIn{$\mathcal{D}_{\mathrm{DT}}=\{\mathcal{P}_u^{\mathrm{DT}}\}_{u=1}^{N}$; path-power calibration measurements $\mathcal A^{(t-1)}$ available before slot $t$; $N$ candidate users}
\KwOut{$\overline{\mathbf{U}}^{(t)},\mathbf{V}^{(t)},\mathcal K^{(t)}$}
\textbf{Slot $t$:} calibrate the DT path powers using only $\mathcal A^{(t-1)}$\;
\For{$u\in\{1,\ldots,N\}$}{
Form $\hat{\mathbf{R}}_u^{(t)}$ from the corrected path set\;
}
Set $\hat{\mathbf{R}}^{(t)}=N^{-1}\sum_u\hat{\mathbf{R}}_u^{(t)}$ and let $\overline{\mathbf{U}}^{(t)}$ contain its $r$ dominant eigenvectors\;
\For{$u\in\{1,\ldots,N\}$}{
Compute $s_u^{(t)}=\mathrm{tr}((\overline{\mathbf{U}}^{(t)})^{\sf H}\hat{\mathbf{R}}_u^{(t)}\overline{\mathbf{U}}^{(t)})$\;
}
\textbf{Prescreen:} select top-$N_s$ users $\mathcal{S}^{(t)}$ by $s_u^{(t)}$ (Def.~\ref{def:dt-prescreen})\;
\For{$u\in\mathcal{S}^{(t)}$}{
User estimates $\hat{\mathbf{g}}_u^{(t)}=\hat{\mathbf{h}}_u^{\sf H}\overline{\mathbf{U}}^{(t)}$ from the beamformed reference signal\;
Compute $b_u$, quasi-SINR $Q_u$, and the one-bit cone-admissibility indicator $c_u$; report $(b_u,Q_u,c_u)$\;
}
\textbf{Schedule:} per beam $i$, pick the cone-qualified user with largest $Q$; keep at most the top-$K$ winners and fill remaining positions by descending $Q$ over the other reports to obtain $\mathcal K^{(t)}$\;
\textbf{Effective-CSI acquisition and precoding:} acquire effective CSI from $\mathcal K^{(t)}$; compute $\mathbf V^{(t)}$ by WMMSE\;
\textbf{Optional anchor update:} after scheduling, add separate path-power calibration measurements from eligible users in $\mathcal K^{(t)}$ and make them available beginning in slot $t+1$\;
\Return $\overline{\mathbf{U}}^{(t)},\mathbf{V}^{(t)},\mathcal K^{(t)}$\;
\end{algorithm}

\subsection{Log-det Variant: DiTUS-L}
We introduce \textbf{DiTUS-L} (Log-det) as a covariance-aware alternative to DiTUS-P. It replaces projection-score ranking with greedy pure log-det prescreening while retaining the same two-stage feedback protocol: the $N_s$ shortlisted users each return one scalar report $(b_u,Q_u,c_u)$, the BS selects the final $K$ users from those reports, and only those $K$ users subsequently provide effective-channel vectors for inner precoding.

\subsubsection{Step 1: DT pure log-det prescreening}
Let $\overline{\mathbf{U}}\in\mathbb{C}^{M\times r}$ be the DT reference-beam matrix and define the beam-domain covariance of user $u$:
\begin{equation}
\mathbf{C}_u \triangleq \overline{\mathbf{U}}^{\sf H}\hat{\mathbf{R}}_u\overline{\mathbf{U}}\in\mathbb{C}^{r\times r}.
\end{equation}
To limit complexity for large $N$, DiTUS-L restricts greedy optimization to the fixed candidate set
\begin{equation}
\mathcal V_{\rm cand}
\triangleq \operatorname{Top}_{N_{\rm cand}}(\{1,\ldots,N\};s_u),
\qquad
N_{\rm cand}=\min\{N,\CfgCandidateMult N_s\},
\label{eq:logdet-candidate-set}
\end{equation}
using the projection score in~\eqref{eq:prescreen-score}. DiTUS-L then solves the aggregate-covariance log-det selection problem
\begin{equation}\label{eq:new-logdet-objective}
\max_{\mathcal{T}\subseteq\mathcal V_{\rm cand},\,|\mathcal{T}|\le N_s}
f(\mathcal T),\qquad
f(\mathcal T)\triangleq
\log\det\!\Big(\mathbf{I}_r+\gamma\!\sum_{u\in\mathcal{T}}\mathbf{C}_u\Big),
\end{equation}
where $\gamma=\rho_{\mathrm{tx}}=\rho/K$ and $\rho=P/\sigma^2=10^{\mathrm{SNR}/10}$; the simulations normalize $\sigma^2$ to one. If the retained users are viewed as a virtual multiple-access channel (MAC) with beam-domain covariances $\{\gamma\mathbf C_u\}$, $f$ is the corresponding log-det rate expression~\cite{Cover:2006}; maximizing it favors jointly strong covariances that span distinct spatial directions. The physical-channel and DT-covariance powers use the same pool-level normalization. Accordingly, the common rescaling $\mathbf C_u'=c\mathbf C_u$, $\gamma'=\gamma/c$, leaves $f$ unchanged.

\begin{lemma}[Monotone submodularity of the log-det objective]\label{prop:logdet}
For fixed $\overline{\mathbf U}$, $\gamma\ge0$, and retained set $\mathcal V_{\rm cand}$, $f$ is normalized, monotone, and submodular. Let $\mathcal S_{\rm g}$ be the size-$N_s$ greedy set and $\mathcal S_{\rm cand}^{\star}$ the optimum of~\eqref{eq:new-logdet-objective}. Then
\begin{equation}
f(\mathcal S_{\rm g})
\ge\left[1-\left(1-\frac{1}{N_s}\right)^{N_s}\right]
f(\mathcal S_{\rm cand}^{\star})
\ge(1-1/e)f(\mathcal S_{\rm cand}^{\star}).
\end{equation}
This is the standard greedy cardinality bound for monotone submodular maximization~\cite{Nemhauser:1978,Shamaiah:2010}.
\end{lemma}
\begin{IEEEproof}
See Appendix~\ref{app:lem1}.
\end{IEEEproof}
The guarantee is relative to the optimum on the fixed retained set $\mathcal V_{\rm cand}$. It becomes a full-pool guarantee only when $N_{\rm cand}=N$; projection-score truncation itself has no approximation ratio here. It also concerns the predicted-covariance objective $f$, not the subsequent CQI selection, WMMSE sum rate, or DT mismatch.

Greedy construction uses the marginal gain
\begin{align}
\Delta_{\mathrm{ld}}(u\mid\mathcal{B})
&=\log\det\!\Big(\mathbf{I}_r+\gamma(\mathbf{\Sigma}_{\mathcal{B}}+\mathbf{C}_u)\Big)
-\log\det\!\Big(\mathbf{I}_r+\gamma\mathbf{\Sigma}_{\mathcal{B}}\Big),
\label{eq:new-logdet-marginal}
\end{align}
where $\mathbf{\Sigma}_{\mathcal{B}}=\sum_{v\in\mathcal{B}}\mathbf{C}_v$. The greedy rule evaluates this exact marginal gain and, because $f$ is monotone, selects exactly $N_s$ users whenever $|\mathcal V_{\rm cand}|\ge N_s$.

\subsubsection{Step 2: Scalar-CQI selection and final effective-CSI acquisition}
Given the prescreened set $\mathcal{S}$, each user $u\in\mathcal{S}$ locally forms its estimated effective-channel vector
\begin{equation}
\hat{\mathbf{g}}_u=\hat{\mathbf{h}}_u^{\sf H}\overline{\mathbf{U}},
\end{equation}
and reports only
\begin{align}\label{eq:cqi}
b_u &\triangleq \arg\max_{i\in\{1,\ldots,r\}}\!|\hat{g}_{u,i}|^2,
&Q_u &\triangleq \frac{\|\hat{\mathbf{g}}_u\|^2}{\Delta},\nonumber\\
c_u &\triangleq\mathbf{1}\!\left\{\frac{|\hat{g}_{u,b_u}|^2}{\|\hat{\mathbf{g}}_u\|^2}\ge\alpha_{\rm cone}\right\}.&&
\end{align}
Thus, $\hat{\mathbf g}_u$ remains local to the user. The BS first forms the cone-qualified set $\mathcal{W}_i=\{u\in\mathcal{S}:b_u=i,\ c_u=1\}$ for each beam and its per-beam winner pool
\begin{equation}
\mathcal{G}=\Big\{\arg\max_{u\in\mathcal{W}_i}Q_u:\ i=1,\ldots,r,\ \mathcal{W}_i\neq\emptyset\Big\}.
\end{equation}
When $r>K$ or fewer than $K$ cones are populated, the BS retains at most the $K$ strongest per-beam winners, $\mathcal{K}_0=\operatorname{Top}_{\min(K,|\mathcal G|)}(\mathcal G;Q_u)$, and fills any remaining positions using only the reported scalar $Q_u$:
\begin{equation}\label{eq:new-fill}
\mathcal{K}
\!=\!\mathcal{K}_{0}
\cup
\operatorname{Top}_{K-|\mathcal{K}_{0}|}
\!\Big(\mathcal{S}\!\setminus\!\mathcal{K}_{0};\ Q_u\Big).
\end{equation}
Because every shortlisted user reports $(Q_u,c_u)$, the BS need not acquire any shortlisted user's effective-channel vector before final selection. After $\mathcal K$ is fixed, the BS acquires $r$-dimensional effective-channel vectors only from those $K$ users and computes the WMMSE inner precoder described in Section~\ref{sec:second-stage}.

\subsection{Computational Complexity Comparison}
We compare the asymptotic computational complexity of the proposed variants against conventional SUS:
\begin{itemize}
    \item \textbf{SUS}: Using iterative Gram--Schmidt residual updates, SUS requires $O(LK^2M)$ projection work over at most $K$ selections and $L$ observed candidates, plus $O(LKM)$ correlation filtering.
    \item \textbf{DiTUS-P}: For an average of $P_{\rm path}$ DT paths, projection-score evaluation costs $O(NP_{\rm path}Mr)$ and sorting costs $O(N\log N)$. User-side reference-beam projection costs $O(N_sMr)$ in aggregate; the BS receives only scalar reports.
    \item \textbf{DiTUS-L}: Forming the beam-domain covariances $\mathbf C_u$ has the same $O(NP_{\rm path}Mr)$ order. After the projection-score prefilter in~\eqref{eq:logdet-candidate-set}, direct pure log-det greedy evaluation over $N_{\rm cand}$ covariance candidates costs $O(N_sN_{\rm cand}r^3)$. DiTUS-L requires no effective-CSI-based refinement across the shortlist.
\end{itemize}
These expressions apply to the algorithms defined above and do not imply a universal linear-versus-quadratic advantage in $M$.

\section{Extension to Proportional Fairness}
While \eqref{eq:optimization} targets instantaneous sum-rate, practical systems often balance throughput with user fairness via the PF criterion~\cite{Viswanath:2002}. Maximizing the PF utility is asymptotically equivalent to maximizing a weighted sum rate in each slot:
\begin{equation}\label{eq:weighted-sr-problem}
    \underset{\mathcal K_t,\mathbf F_{\mathcal K_t}}{\text{maximize}} \quad \sum_{u\in\mathcal K_t} w_u(t) \log_2(1+\gamma_u(t)),
\end{equation}
where $\bar R_u(0)=\epsilon_{\rm PF}=\CfgPfEpsilon$ and $w_u(t)=1/\max\{\bar R_u(t-1),\epsilon_{\rm PF}\}$ is computed from the throughput state available before slot $t$. A common positive scaling normalizes the weights to unit mean without changing their relative values or the maximizing user set and precoder.
The PF extension evaluated here is based on DiTUS-P and incorporates $w_u(t)$ into both stages. It uses $s_u^{\mathrm{PF}}(t)=w_u(t)\operatorname{tr}(\overline{\mathbf U}^{\sf H}\hat{\mathbf R}_u\overline{\mathbf U})$, thereby favoring users with low accumulated throughput. Each shortlisted user still reports only the unweighted scalar tuple $(b_u,Q_u,c_u)$, and the BS ranks reports by $w_u(t)Q_u$ using the weights available at the beginning of slot $t$. Thus, no weight or shortlisted-user effective-channel vector must be reported. \rev{After $\mathcal K_t$ is selected and its effective CSI is acquired, weighted WMMSE computes a stationary point of $\sum_{k\in\mathcal K_t}w_k(t)\log_2(1+\gamma_k)$ by setting $\nu_k=w_k(t)/e_k$ in~\eqref{eq:wmmse-update}~\cite{Shi:2011}.}

More generally, a static factor $w_u^{\mathrm{app}}$ can encode application-class priority. The same weighted score and CQI ranking then provide a priority-aware heuristic, but without explicit queue or rate constraints they do not guarantee latency or minimum rate. The time-domain PF procedure is summarized in Algorithm~\ref{alg:pf}.

\begin{algorithm}[t]
\caption{DiTUS with Proportional Fairness (Time-Domain)}
\label{alg:pf}
\SetAlgoLined
\KwIn{$\mathcal U\triangleq\{1,\ldots,N\}$, $\overline{\mathbf U}$, $\{\hat{\mathbf R}_u\}$, $N_s$, $K$, $T_c$, $T$, and $\epsilon_{\rm PF}$; initialize $\bar R_u(0)\gets\epsilon_{\rm PF}$}
\For{$t=1$ \KwTo $T$}{
$w_u(t)\gets 1/\max\{\bar R_u(t-1),\epsilon_{\rm PF}\}$; normalize the common scale; compute $s_u^{\mathrm{PF}}(t)=w_u(t)\,\mathrm{tr}(\overline{\mathbf{U}}^{\sf H}\hat{\mathbf{R}}_u\overline{\mathbf{U}})$\;
Select top-$N_s$ users $\mathcal{S}_t$ by $s_u^{\mathrm{PF}}(t)$\;
Each $u\in\mathcal{S}_t$ estimates $\hat{\mathbf{g}}_u$ and reports $(b_u,Q_u,c_u)$\;
$\mathcal{K}_t\gets\operatorname{Select}_K(\mathcal{S}_t;\,b_u,w_u(t)Q_u,c_u)$; collect effective CSI only from $\mathcal K_t$ and compute weighted-WMMSE $\mathbf{V}$ with the same $w_u(t)$\;
$R_u(t)\gets\log_2(1+\gamma_u(t))$ for $u\in\mathcal K_t$ and $R_u(t)\gets0$ otherwise; $\bar R_u(t)\gets(1-1/T_c)\bar R_u(t-1)+R_u(t)/T_c$\;
}
\end{algorithm}

$\operatorname{Select}_K$ denotes Algorithm~\ref{alg:ditus}'s per-beam winner selection and CQI-based completion rule with $Q_u$ replaced by $w_u(t)Q_u$. The PF evaluation omits GP updates; joint PF--GP adaptation is not considered. Incorporating PF weights into both stages balances accumulated throughput against instantaneous spatial quality, and the resulting fairness is assessed empirically over time.

\section{Simulation Results}
\rev{Downlink reference-signal resources and uplink feedback bits are reported separately because combining them requires an implementation-specific mapping. Define $B_s(r)\triangleq\lceil\log_2r\rceil+B_Q+1$ for one scalar report and $B_g(d)$ for a quantized $d$-dimensional channel vector. Excluding anchor costs, the feedback loads are}
\begin{align}\label{eq:overhead-resource}
 B_{\rm full}&=N B_g(M), & B_{L\text{-SUS}}&=L B_g(M),\\
 B_{\rm DiTUS}&=N_sB_s(r)+KB_g(r).&&
\end{align}
\rev{DiTUS uses the common downlink cost $C_{\rm RS}(r)$; under the same port-count model, the SUS and full-pool baselines acquire $M$-dimensional full CSI and use $C_{\rm RS}(M)$. If $n_{\rm A}$ new path-resolved anchor measurements are acquired, their measurement and feedback costs $n_{\rm A}C_{\rm A}$ and $n_{\rm A}B_{\rm A}$ are reported separately. For the large-pool configuration $N=\CfgBaseN$ and $K=\CfgBaseK$ in Table~\ref{tab:simulation-config}, the \emph{effective-CSI reporter count} relative to full-pool acquisition decreases by}
\begin{equation}\label{eq:overhead-reduction}
    1-\frac{K}{N} = 1 - \frac{\CfgBaseK}{\CfgBaseN} = 96.8\%.
\end{equation}
\rev{This is a reporter-count, not total-overhead, reduction because the $N_s$ scalar reports and anchor resources remain. We therefore report the tuple $(C_{\rm RS},N_sB_s(r),KB_g(r),n_{\rm A}(C_{\rm A},B_{\rm A}))$ rather than rate versus unspecified total bits.}

\petar{Table~\ref{tab:overhead} extends this accounting to every compared scheme. All prescreening methods feeding the shared second stage---DiTUS-P/L, Max-RSRP, Max-Power, and Random (DT/DFT)---incur the identical per-instance cost $C_{\rm RS}(r)$, $N_s$ scalar reports, and $K$ effective-channel reports and differ only in the ranking rule, so comparisons among them are equal-overhead. Their ranking inputs, beam-wise long-term RSRP and total long-term path power, are read from the same perturbed DT statistics DiTUS uses and are charged no per-slot cost. Without a DT they would require slow-timescale RSRP reporting from all $N$ candidates, a cost we note but do not charge, so the comparison is conservative for DiTUS. By contrast, $L$-user SUS acquires $M$-dimensional CSI and is not overhead-matched at a common $L=N_s$ since $B_g(M)\gg B_s(r)$.}

\petar{The ``Full-pool perfect-CSI SUS'' curve is an oracle reference, not a deployable scheme: it runs SUS over all $N$ candidates with perfect, unquantized $M$-dimensional CSI and is charged no acquisition cost. Its feedback load is therefore unbounded rather than $NB_g(M)$, since perfect CSI corresponds to $B_g(M)\to\infty$. It measures the value of unrestricted, error-free channel knowledge for this selector---not optimal performance, as SUS is greedy---so gaps to it reflect the price of a finite acquisition budget.}

\begin{table*}[t]
\centering
\caption{\petar{Per-instance acquisition accounting for all compared schemes (anchor resources excluded and reported separately). $B_s(r)=\lceil\log_2r\rceil+B_Q+1$ is one scalar report and $B_g(d)$ one quantized $d$-dimensional channel vector. ``Offline DT'' ranking inputs are read from the perturbed DT statistics and incur no per-slot air-interface cost.}}
\label{tab:overhead}
\setlength{\tabcolsep}{4pt}
\renewcommand{\arraystretch}{1.15}
\scriptsize
\begin{tabular}{l|c|c|c|c|c}
\hline
\textbf{Scheme} &
\petar{\textbf{Ranking-stage input}} &
\textbf{CSI-reporting users} &
\textbf{Scalar-report users} &
\petar{\textbf{Reference signal}} &
\textbf{Feedback bits} \\
\hline
\petar{Full-pool perfect-CSI SUS (oracle)} &
\petar{---} &
\petar{$N$ ($M$-dim, unquantized)} &
\petar{$0$} &
\petar{$C_{\rm RS}(M)$} &
\petar{unbounded ($B_g(M)\!\to\!\infty$)} \\
\hline
Full-pool CSI &
\petar{---} &
$N$ \petar{($M$-dim)} &
$0$ &
\petar{$C_{\rm RS}(M)$} &
$NB_g(M)$ \\
\hline
$L$-user SUS &
\petar{---} &
$L$ \petar{($M$-dim)} &
$0$ &
\petar{$C_{\rm RS}(M)$} &
$LB_g(M)$ \\
\hline
DiTUS-P/L &
\petar{Offline DT covariance} &
$K$ \petar{($r$-dim)} &
$N_s$ &
\petar{$C_{\rm RS}(r)$} &
$N_sB_s(r)+KB_g(r)$ \\
\hline
\petar{Max-RSRP} &
\petar{Offline DT beam-wise RSRP} &
\petar{$K$ ($r$-dim)} &
\petar{$N_s$} &
\petar{$C_{\rm RS}(r)$} &
\petar{$N_sB_s(r)+KB_g(r)$} \\
\hline
\petar{Max-Power} &
\petar{Offline DT path power} &
\petar{$K$ ($r$-dim)} &
\petar{$N_s$} &
\petar{$C_{\rm RS}(r)$} &
\petar{$N_sB_s(r)+KB_g(r)$} \\
\hline
\petar{Random (DT/DFT)} &
\petar{None} &
\petar{$K$ ($r$-dim)} &
\petar{$N_s$} &
\petar{$C_{\rm RS}(r)$} &
\petar{$N_sB_s(r)+KB_g(r)$} \\
\hline
\end{tabular}
\end{table*}

\subsection{Simulation Setup}
\subsubsection{Instantaneous vs. Statistical Channel Generation}
We distinguish between (i) \emph{statistical} channel information used by the DT and (ii) \emph{instantaneous} channels used to evaluate scheduling and precoding.
In the simulations, DT statistics comprise ray-tracing AoDs, path powers, and delays. DiTUS uses the AoDs and powers to construct $\hat{\mathbf{R}}_u$, $\hat{\mathbf{R}}$, and $s_u$ through~\eqref{eq:corrected-cov}--\eqref{eq:cell-cov}; delay does not enter the narrowband covariance.

Instantaneous narrowband channels follow~\eqref{eq:geo-channel}, with i.i.d. phases $\phi_{u,\ell}\sim\mathrm{Unif}[-\pi,\pi]$ resampled for each realization. The physical channel uses the unperturbed ray-tracing paths $(L_u,p_{u,\ell},\theta_{u,\ell})$, whereas the scheduler uses their specified DT perturbations $(\hat L_u,\hat p_{u,\ell},\hat\theta_{u,\ell})$, with path powers further corrected when GP calibration is enabled. The physical channels and DT covariance powers are normalized by the same pool-dependent factor, preserving their relative scale.
To emulate imperfect channel knowledge, we first perturb instantaneous channels via an additive estimation-error model, and then compute first-stage effective CQI by projecting the noisy channel estimates onto the reference beams.

\subsubsection{Parameters}
\rev{We evaluate DiTUS using the DeepMIMO ASU Campus 3.5-GHz scenario~\cite{alkhateeb2019deepmimo}. In each Monte Carlo trial, the candidate pool is sampled from active users whose total long-term received power exceeds $\CfgMinRsrp$~dB. Unless stated otherwise, the large-pool studies use $N=\CfgBaseN$; the main SNR study uses $N=\CfgMainN$, which permits a full-pool perfect-CSI SUS reference over all $N$ users.}
Unless stated otherwise, we use $M=\CfgBaseM$ BS antennas, schedule $K=\CfgBaseK$ users per slot, set $r=K=\CfgBaseRank$, use $\alpha_{\rm cone}=\CfgAlphaCone$ and $\alpha_{\rm SUS}=\CfgAlphaSUS$, and retain at most $N_{\rm cand}=\min\{N,\CfgCandidateMult N_s\}$ covariance candidates for DiTUS-L. DiTUS-L uses the pure log-det objective in~\eqref{eq:new-logdet-objective}. Table~\ref{tab:simulation-config} summarizes the common simulation parameters and experiment-specific settings.
\rev{To model imperfect CSI, we use $\hat{\mathbf{h}}=\sqrt{1-\zeta}\,\mathbf{h}+\sqrt{\zeta}\,\mathbf{e}$, where $\mathbf e\sim\mathcal{CN}(\mathbf0,\mathbf I_M/M)$ is independent of $\mathbf h$. The physical channels are normalized to unit average energy over each candidate pool, and $\mathbb E[\|\mathbf e\|^2]=1$; thus, $\zeta\in[0,1]$ controls the average channel-error mixture and is distinct from the transmit SNR $\rho$ in Section~IV. We set $\zeta=\CfgCsiError$.}
For the main SNR sweep in Fig.~\ref{fig:dt_vs_conventional_snr}, we evaluate $\{-5,0,5,10,15,20,25\}$~dB with $N=\CfgMainN$ and effective-CSI acquisition budget $L=\CfgMainL$. The prescreening size is $N_s=\CfgMainNs=L$ in this sweep. We separately sweep $L\in\{16,24,32,48,64,80,128\}$ at 15~dB.

\begin{table}[!t]
\centering
\caption{Simulation parameters and experiment-specific settings.}
\label{tab:simulation-config}
\setlength{\tabcolsep}{3.5pt}
\renewcommand{\arraystretch}{1.05}
\scriptsize
\begin{tabular}{l|l}
\hline
\textbf{Item} & \textbf{Setting} \\
\hline
Base system & $N=\CfgBaseN$, $M=\CfgBaseM$, $K=\CfgBaseK$ \\
Base reporting & $L=\CfgBaseL$, $N_s=\CfgBaseNs$, $r=\CfgBaseRank$ \\
Main SNR sweep & $N=\CfgMainN$, $\{-5,0,5,10,15,20,25\}$~dB, \CfgTrials\ trials/point \\
Selection & $\alpha_{\rm cone}=\CfgAlphaCone$, $\alpha_{\rm SUS}=\CfgAlphaSUS$, $N_{\rm cand}\leq\CfgCandidateMult N_s$ \\
Main DT/CSI errors & $\sigma_{\rm aod}=\CfgMainAodError^\circ$, $\sigma_{\rm pow}=\CfgMainPowerError$~dB, $\zeta=\CfgCsiError$ \\
Inner precoder & WMMSE, \CfgWmmseIters\ iterations, tolerance $\CfgWmmseTol$ \\
PF experiment & $N_s=\CfgPfNs$, $T=\CfgPfSlots$, $T_c=\CfgPfTc$, \CfgPfTrajectories\ trajectories \\
Causal-GP experiment & \CfgGpSlots\ slots, \CfgGpTrajectories\ trajectories, up to \CfgGpAnchorsPerSlot\ calibrated users/slot \\
\hline
\end{tabular}
\end{table}

\subsubsection{Monte Carlo Methodology and Uncertainty Quantification}
For each trial or trajectory, we draw a candidate pool, instantaneous path phases, DT perturbations, and a channel-estimation error realization. All methods compared within that trial share the same physical pool and channel/error realization, whereas random choices made by individual methods are independent. For paired parameter sweeps, the underlying realization is reused across the swept settings. Error bars report two-sided 95\% confidence-interval half-widths over independent trials or trajectories, not over users or time slots: Student-$t$ intervals over 200 trials for the $L/N_s$, robustness, $M/K$, and cone-threshold studies, and normal-approximation intervals $1.96s/\sqrt{n}$ for the main-SNR study and the 20-trajectory PF and causal-GP studies. No correction is applied for multiple comparisons across sweep points; accordingly, isolated small differences are interpreted as exploratory.

\subsection{Compared Schemes}
\label{sec:compared-schemes}
\rev{Methods in the two-stage framework use $\mathbf{F}_{\mathrm{ts}}=\overline{\mathbf{U}}\mathbf{V}$, where $\overline{\mathbf{U}}$ is the DT- or discrete-Fourier-transform (DFT)-based outer beamformer and $\mathbf{V}$ is the inner WMMSE precoder designed from $\hat{\mathbf{h}}_k^{\sf H}\overline{\mathbf{U}}$. The SUS baseline instead acquires full-dimensional CSI from its $L$ granted users and applies ZF with water-filling in the antenna domain. The proposed schemes and benchmarks are summarized as follows:}
\begin{itemize}
    \item \textbf{DiTUS-P (Proposed):} DT projection-score prescreening~\eqref{eq:prescreen-score}, cone-based CQI selection, and DT-based two-stage precoding.
    \item \textbf{DiTUS-L (Proposed):} Greedy pure log-det prescreening~\eqref{eq:new-logdet-objective}, one scalar-CQI reporting stage, and DT-based two-stage precoding using effective-channel vectors only from the final $K$ users.
    \item \textbf{SUS~\cite{Yoo:2006}:} SUS over $L$ randomly granted users with full-dimensional CSI, ZF, and water-filling in the antenna domain; Fig.~\ref{fig:dt_vs_conventional_snr} also includes ``Full-pool perfect-CSI SUS,'' obtained by running SUS over all $N$ users with perfect full-dimensional CSI. \petar{This curve is an overhead-free oracle, not a competing scheme; the remaining baselines share DiTUS's per-instance overhead and differ only in the ranking rule (Table~\ref{tab:overhead}).}
    \item \textbf{Max-RSRP:} Ranks users by the maximum long-term reference signal received power (RSRP) over the DFT beams, then applies the same DFT-based second-stage selector and precoder.
    \item \textbf{Max-Power:} Ranks users by total long-term path power, then uses the same DFT-based second-stage processing as Max-RSRP.
    \item \textbf{Random (DT/DFT):} Uniform random prescreening followed by the DT- or DFT-based second-stage selector and two-stage precoder.
\end{itemize}

\subsection{Numerical Results}
\rev{Fig.~\ref{fig:dt_vs_conventional_snr} plots the mean sum rate with normal-approximation 95\% confidence intervals over \ProdMainTrials\ trials. Both DiTUS variants outperform $L$-user SUS at every tested SNR. At 15~dB, DiTUS-P and DiTUS-L achieve $\ProdMainPAtFifteen\!\pm\!\ProdMainPCIAtFifteen$ and $\ProdMainLAtFifteen\!\pm\!\ProdMainLCIAtFifteen$~bps/Hz, respectively, versus $\ProdMainSusAtFifteen\!\pm\!\ProdMainSusCIAtFifteen$~bps/Hz for SUS. The DiTUS-L gain is $\ProdMainLGainVsGrantedSusAtFifteen\%$ at 15~dB and $\ProdMainLGainVsGrantedSusAtTwenty\%$ at 20~dB. SUS searches only its $L$ randomly granted users with full-dimensional CSI, whereas DiTUS ranks the full pool and requests effective-channel vectors only from the final $K$. The gain over Random (DT), which changes only the shortlist rule, isolates the benefit of DT ranking. The remaining gap to full-pool perfect-CSI SUS ($\ProdMainAllSusAtFifteen\!\pm\!\ProdMainAllSusCIAtFifteen$~bps/Hz at 15~dB) reflects imperfect DT statistics, scalar reports, and reduced-rank precoding; \petar{because that reference is charged no acquisition cost, the gap measures the price of a finite feedback budget rather than a shortfall against a deployable scheme.} Its larger value at 25~dB is consistent with thermal noise receding while ranking error, subspace truncation, and residual interference remain. DiTUS-P and DiTUS-L are nearly identical when $L=N_s$, indicating that the projection score captures most of the prescreening gain in this setting.}

\begin{figure}[!htbp]
    \centering
    \includegraphics[width=\columnwidth]{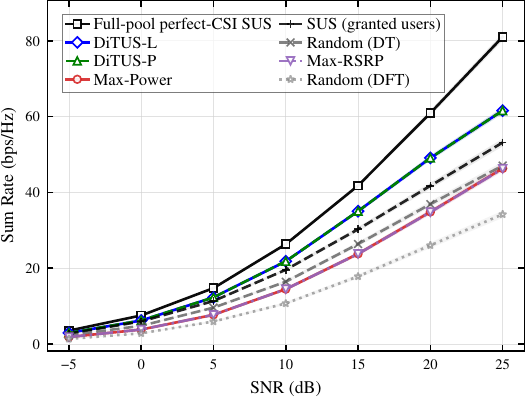}
    \caption{\rev{Sum rate vs. SNR under DT imperfections ($\sigma_{\mathrm{aod}}=\CfgMainAodError^\circ$, $\sigma_{\mathrm{pow}}=\CfgMainPowerError$~dB), with $N=\CfgMainN$, $M=\CfgBaseM$, $K=\CfgBaseK$, $L=N_s=\CfgMainNs$, and $r=\CfgMainRank$. Error bars show normal-approximation 95\% confidence intervals over \CfgTrials\ independent trials. The black ``Full-pool perfect-CSI SUS'' curve runs SUS over all $N$ users with perfect CSI and serves as \petar{an overhead-free, scheduler-specific oracle reference (Table~\ref{tab:overhead})}.}}
    \label{fig:dt_vs_conventional_snr}
\end{figure}

\subsection{Impact of $L$, $N_s$, and Their Decoupling}
Fig.~\ref{fig:combined-L-Ns}(a) couples the effective-CSI acquisition and scalar-report budgets. At the base point $L=N_s=\ProdLBase$, DiTUS-P, DiTUS-L, and SUS achieve $\ProdLBaseP\!\pm\!\ProdLBasePCI$, $\ProdLBaseL\!\pm\!\ProdLBaseLCI$, and $\ProdLBaseSus\!\pm\!\ProdLBaseSusCI$~bps/Hz. Panel~(b) instead holds SUS at $L=64$ users with full-dimensional CSI while every scalar-CQI-based prescreening method uses the swept $N_s$: DiTUS-L changes from $\ProdNsBaseL\!\pm\!\ProdNsBaseLCI$ at $N_s=\ProdNsBase$ to $\ProdNsMaxL\!\pm\!\ProdNsMaxLCI$~bps/Hz at $N_s=\ProdNsMax$. Thus panel~(b) isolates the performance effect and cost of additional scalar-CQI reports, but is not an equal-total-feedback comparison; $N_s=N$ removes DT prescreening while retaining scalar CQI from all users.

\begin{figure}[!htbp]
    \centering
    \includegraphics[width=\columnwidth]{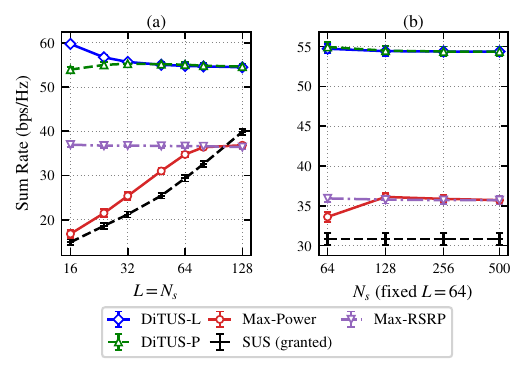}
    \caption{Sum rate at 15~dB with $N=500$, $M=64$, $K=r=16$: (a) coupled $L=N_s$; (b) scalar-report $N_s$ with SUS fixed at $L=64$ users providing full-dimensional CSI. Error bars show Student-$t$ 95\% confidence intervals over 200 paired trials. Panel~(b) gives every scalar-CQI-based method the same $N_s$, increases scalar-feedback bits, and is not an equal-total-feedback comparison.}
    \label{fig:combined-L-Ns}
\end{figure}

\subsection{Proportional Fairness Analysis}
\rev{The PF experiment follows Algorithm~\ref{alg:pf}. Each slot uses the current exponentially weighted moving-average throughput state for weighted top-$N_s$ prescreening; $N_s=\CfgPfNs$ users provide scalar CQI and the final $K=\CfgBaseK$ provide effective-channel vectors. We use $T=\CfgPfSlots$, $N=\CfgPfN$, SNR~$=\CfgPfSnr$~dB, and $T_c=\CfgPfTc$ over \CfgPfTrajectories\ trajectories. After a $5T_c=250$-slot warm-up, the remaining 250 slots determine sum rate, Jain's index, fifth-percentile throughput, and service coverage, the fraction of users scheduled at least once.}

\rev{Across \ProdPFTrajectories\ trajectories, PF-DiTUS obtains mean sum rate $\ProdPFSumRate\!\pm\!\ProdPFSumRateCI$~bps/Hz, Jain index $\ProdPFJain\!\pm\!\ProdPFJainCI$, and service coverage $\ProdPFDitusServiceCoverage\!\pm\!\ProdPFDitusServiceCoverageCI$. Its CDF is lower than those of the three equal-reporting-budget methods at nearly all positive thresholds in Fig.~\ref{fig:pf-simulation}, whereas PF-Random has less probability mass at zero; hence, the CDFs cross. Fifth-percentile throughput is zero for all methods over this horizon, so Table~\ref{tab:pf-summary} reports coverage with trajectory-level intervals.}

\begin{figure}[!htbp]
    \centering
    \includegraphics[width=\columnwidth]{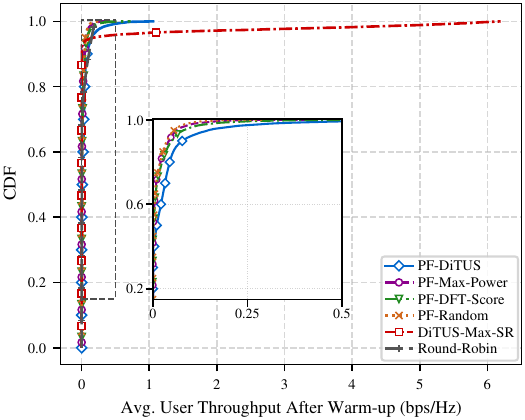}
    \caption{Per-user throughput CDF after a 250-slot warm-up period, pooled across \ProdPFTrajectories\ independent \ProdPFSlots-slot trajectories. The dashed box marks the $0\leq R_u\leq0.5$~bps/Hz domain enlarged in the inset for PF-DiTUS and its three baselines with the same reporting budget. The full view also includes the lower-feedback Round-Robin and unweighted DiTUS-Max-SR references.}
    \label{fig:pf-simulation}
\end{figure}

\begin{figure*}[!t]
    \centering
    \includegraphics[width=\textwidth]{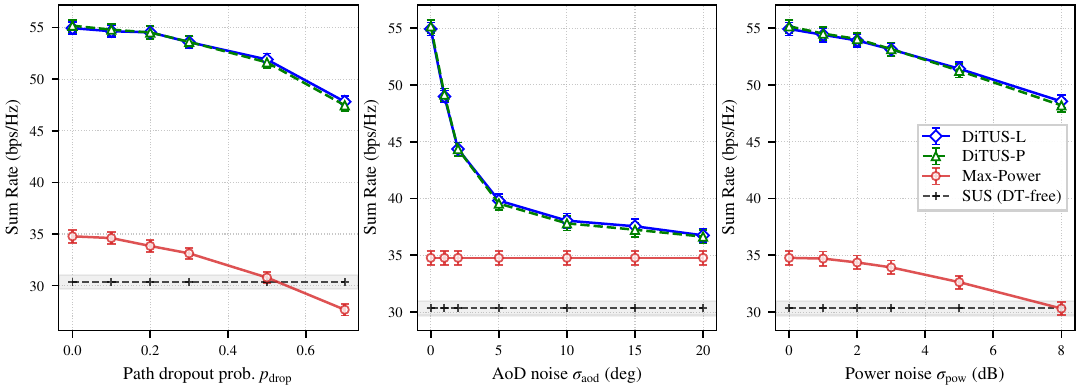}
    \caption{Sum rate under (a) path dropout, (b) AoD noise, and (c) power noise at 15~dB with $N=500$, $M=64$, $K=16$, and $L=N_s=64$. Bars are Student-$t$ 95\% confidence intervals over 200 trials; SUS is DT-independent.}
    \label{fig:dt_imperfection_robustness}
\end{figure*}

\begin{table}[!t]
\centering
\caption{PF comparison at $N=500$, $K=16$, $N_s=256$, and 20~dB. Sum rate is in bps/Hz and service coverage is a fraction; entries are trajectory means $\pm$ 95\% interval half-widths.}
\label{tab:pf-summary}
\renewcommand{\arraystretch}{1.1}
\scriptsize
\setlength{\tabcolsep}{2.4pt}
\begin{tabular}{l|c|c|c}
\hline
\textbf{Method} & \textbf{Sum rate} & \textbf{Jain index} & \textbf{Coverage} \\
\hline
PF-DiTUS & $\ProdPFDitusSumRate\!\pm\!\ProdPFDitusSumRateCI$ & $\ProdPFDitusJain\!\pm\!\ProdPFDitusJainCI$ & $\ProdPFDitusServiceCoverage\!\pm\!\ProdPFDitusServiceCoverageCI$ \\
PF-Max-Power & $\ProdPFMaxPowerSumRate\!\pm\!\ProdPFMaxPowerSumRateCI$ & $\ProdPFMaxPowerJain\!\pm\!\ProdPFMaxPowerJainCI$ & $\ProdPFMaxPowerServiceCoverage\!\pm\!\ProdPFMaxPowerServiceCoverageCI$ \\
PF-DFT-Score & $\ProdPFDftScoreSumRate\!\pm\!\ProdPFDftScoreSumRateCI$ & $\ProdPFDftScoreJain\!\pm\!\ProdPFDftScoreJainCI$ & $\ProdPFDftScoreServiceCoverage\!\pm\!\ProdPFDftScoreServiceCoverageCI$ \\
PF-Random & $\ProdPFRandomSumRate\!\pm\!\ProdPFRandomSumRateCI$ & $\ProdPFRandomJain\!\pm\!\ProdPFRandomJainCI$ & $\ProdPFRandomServiceCoverage\!\pm\!\ProdPFRandomServiceCoverageCI$ \\
DiTUS-Max-SR & $\ProdPFMaxSrSumRate\!\pm\!\ProdPFMaxSrSumRateCI$ & $\ProdPFMaxSrJain\!\pm\!\ProdPFMaxSrJainCI$ & $\ProdPFMaxSrServiceCoverage\!\pm\!\ProdPFMaxSrServiceCoverageCI$ \\
Round-Robin & $\ProdPFRoundRobinSumRate\!\pm\!\ProdPFRoundRobinSumRateCI$ & $\ProdPFRoundRobinJain\!\pm\!\ProdPFRoundRobinJainCI$ & $\ProdPFRoundRobinServiceCoverage\!\pm\!\ProdPFRoundRobinServiceCoverageCI$ \\
\hline
\end{tabular}
\end{table}

\subsection{Robustness to DT Imperfections}
\rev{Fig.~\ref{fig:dt_imperfection_robustness} varies one DT-error component at a time. At the largest tested path-drop probability $p_{\rm drop}=\ProdDropMax$, AoD standard deviation $\sigma_{\rm aod}=\ProdAodMax^\circ$, and power-error standard deviation $\sigma_{\rm pow}=\ProdPowerMax$~dB, DiTUS-L obtains $\ProdDropMaxL\!\pm\!\ProdDropMaxLCI$, $\ProdAodMaxL\!\pm\!\ProdAodMaxLCI$, and $\ProdPowerMaxL\!\pm\!\ProdPowerMaxLCI$~bps/Hz, respectively. Across all three sweeps, the dominant trend is a gradual rate decrease as DT error increases; isolated nonmonotonic variations are contained within the reported confidence intervals.}

\subsection{Cone-Threshold Sensitivity}
Fig.~\ref{fig:alpha-cone-sensitivity} shows a paired 200-trial threshold sweep. DiTUS-L gives $\ProdAlphaMinL\!\pm\!\ProdAlphaMinLCI$, $\ProdAlphaDefaultL\!\pm\!\ProdAlphaDefaultLCI$, and $\ProdAlphaMaxL\!\pm\!\ProdAlphaMaxLCI$~bps/Hz at $\alpha_{\rm cone}=\ProdAlphaMin$, $\ProdAlphaDefault$, and $\ProdAlphaMax$, respectively. At the default threshold, $\ProdAlphaDefaultCone\!\pm\!\ProdAlphaDefaultConeCI$ shortlisted users satisfy the cone condition on average before the CQI-based completion rule is applied. The interior choice $\CfgAlphaCone$ provides a compromise between cone-qualified-user availability and angular selectivity.
\begin{figure}[!htbp]
    \centering
    \includegraphics[width=0.94\columnwidth]{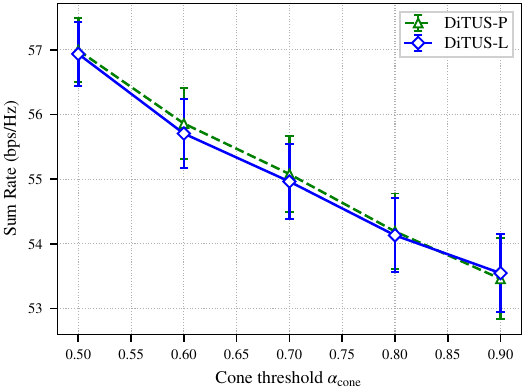}
    \caption{Sum rate versus $\alpha_{\rm cone}$ at 15~dB with $N=500$, $M=64$, $K=16$, and $L=N_s=64$. Bars are Student-$t$ 95\% confidence intervals over 200 trials.}
    \label{fig:alpha-cone-sensitivity}
\end{figure}

\subsection{Causal GP Calibration}\label{sec:causal-gp}
This is a separate systematic-error stress test: the DT path power is increased by 20~dB for every valid path with post-perturbation AoD $-150^\circ<\hat\theta<-120^\circ$, while random path dropout and AoD/power errors are disabled. The causal-GP experiment uses \ProdGPTrajectories\ independent \ProdGPSlots-slot trajectories and collects calibration data from up to \CfgGpAnchorsPerSlot\ scheduled users after each slot; at most 64 anchor users from the previous \CfgGpWindowSlots\ slots are retained. Over slots 51--100, the no-GP method, scheduled-user causal GP, and a random full-pool oracle-label GP reference with the same measured-user budget achieve mean rates $\ProdGPNoRate\!\pm\!\ProdGPNoRateCI$, $\ProdGPCausalRate\!\pm\!\ProdGPCausalRateCI$, and $\ProdGPOracleRate\!\pm\!\ProdGPOracleRateCI$~bps/Hz, respectively. All methods start without labels and coincide in slot~1; labels acquired after slot $t$ first affect slot $t+1$. The full-pool reference draws users from the entire candidate pool and is unavailable under the proposed reporting protocol.

\begin{figure}[!htbp]
    \centering
    \includegraphics[width=\columnwidth]{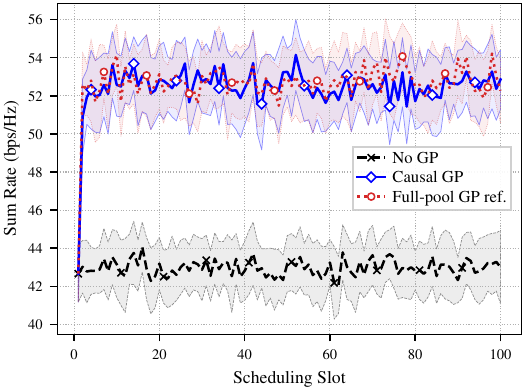}%
    \caption{Slot-wise sum rate at 15~dB for the separate 20-dB sector-bias stress test ($-150^\circ<\hat\theta<-120^\circ$) with $N=500$, $M=64$, $K=16$, and $L=N_s=64$. Across \CfgGpTrajectories\ independent trajectories, curves show slot-wise means and shaded two-sided normal 95\% confidence intervals. Both GP methods collect calibration data from up to \CfgGpAnchorsPerSlot\ users after each slot but differ in scheduled-only versus random full-pool oracle-label access.}
    \label{fig:gp-causal}
\end{figure}

\subsection{Impact of Antenna Count, Scheduled-User Count, and Subspace Rank}
Fig.~\ref{fig:combined-MK-sweep} separates fixed-rank operation from the energy-capture rule in Remark~\ref{rem:rank}. At $M=\ProdMBase$, fixed $r=16$ captures $\ProdMBaseEta\!\pm\!\ProdMBaseEtaCI$ of the predicted covariance energy, while the empirical distribution of $r_{0.9}$ has 5th/median/95th percentiles $\ProdMBaseRankPZeroFive/\ProdMBaseRankMedian/\ProdMBaseRankPNinetyFive$. DiTUS-L obtains $\ProdMBaseL\!\pm\!\ProdMBaseLCI$ with fixed rank and $\ProdMBaseLAdaptive\!\pm\!\ProdMBaseLAdaptiveCI$~bps/Hz with adaptive rank; this comparison changes the reference-signal and effective-CSI dimensions and therefore does not preserve the feedback overhead. Panel~(b) uses $L(K)=\max(64,4K)$, fixed $N_s=64$, and $r=K$; at $K=\ProdKBase$, DiTUS-L obtains $\ProdKBaseL\!\pm\!\ProdKBaseLCI$~bps/Hz.

\begin{figure}[!htbp]
    \centering
    \includegraphics[width=\columnwidth]{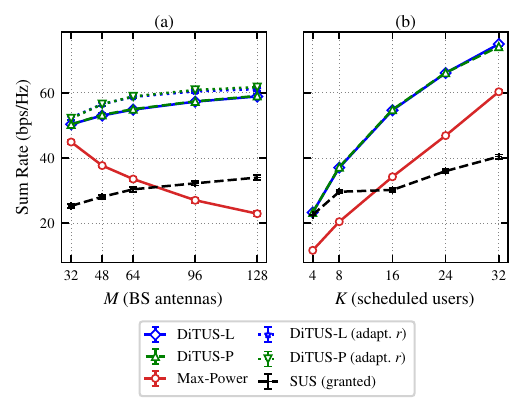}
    \caption{Dimension sweeps at 15~dB with 95\% trial-level confidence intervals: (a) $M$ at $K=16$, $L=N_s=64$, comparing fixed $r=16$ and adaptive $r$; (b) $K$ at $M=64$, $N_s=64$, $r=K$, and $L(K)=\max(64,4K)$. The adaptive-$r$ and $K$-dependent-$L$ cases do not use equal feedback budgets.}
    \label{fig:combined-MK-sweep}
\end{figure}

\section{Conclusion}
\rev{This paper proposed DiTUS, which converts slowly varying DT multipath statistics into projection-energy or log-determinant prescreening before two-level feedback and WMMSE precoding. Its FDD resource model separates $N_s$ scalar reports, $K$ scheduled-user effective-CSI reports, reference signals, and calibration anchors. Across the evaluated ray-tracing settings, DiTUS quantifies instantaneous-CSI acquisition, fairness, rank, and DT-mismatch tradeoffs. At 15~dB with $N=128$, DiTUS-L achieves $\ProdMainLAtFifteen\!\pm\!\ProdMainLCIAtFifteen$~bps/Hz versus $\ProdMainSusAtFifteen\!\pm\!\ProdMainSusCIAtFifteen$~bps/Hz for $L$-user SUS; scheduled-user causal GP calibration also recovers most of the full-pool reference gain.}

\rev{Limitations include the absence of a specified $B_g(r)$ quantizer and rate-versus-total-feedback-bit curve, reliance on slow-timescale identity, location, and path association, and evaluation in one narrowband single-cell ray-tracing scenario. Future work will address delay-aware wideband prescreening, uncertainty-aware calibration, and multi-cell or cell-free operation with adaptive $N_s$, $L$, and $r$.}

\appendices

\section{Proof of Proposition~\ref{prop:rate-gap}}
\label{app:prop2}
By definition, the two surrogate rates are $R_{\rm full}^{\rm IF}=\sum_k\log_2(1+\rho_{\mathrm{tx}}\|\mathbf h_k\|^2)$ and $R_{\rm sub}^{\rm IF}=\sum_k\log_2(1+\rho_{\mathrm{tx}}\|\tilde{\mathbf h}_k\|^2)$. Their difference decomposes exactly per user. Therefore, the bound below applies to the interference-free surrogate rates and does not bound interference-coupled WMMSE rates.

Using the effective-channel vector defined in~\eqref{eq:effective-csi-vector}, decompose $\mathbf{h}_k = \overline{\mathbf{U}}\tilde{\mathbf{h}}_k + \mathbf{h}_k^{\perp}$, where $\mathbf{h}_k^{\perp} = (\mathbf{I}_M-\overline{\mathbf{U}}\overline{\mathbf{U}}^{\mathsf{H}})\mathbf{h}_k$. Since $\overline{\mathbf{U}}$ has orthonormal columns, $\overline{\mathbf{U}}\overline{\mathbf{U}}^{\mathsf{H}}$ is an orthogonal projector, giving the Pythagorean identity $\|\mathbf{h}_k\|^2 = \|\tilde{\mathbf{h}}_k\|^2 + \|\mathbf{h}_k^{\perp}\|^2$. Any two-stage $\mathbf{F}=\overline{\mathbf{U}}\mathbf{V}$ satisfies $\mathbf{h}_k^{\mathsf{H}}\mathbf{F} = \tilde{\mathbf{h}}_k^{\mathsf{H}}\mathbf{V}$, so $\mathbf{h}_k^{\perp}$ contributes no beamforming gain.

For each user, the surrogate-rate gap satisfies
\begin{equation*}
\Delta R_k = \log_2\!\left( \frac{1 + \rho_{\mathrm{tx}} \|\mathbf{h}_k\|^2}{1 + \rho_{\mathrm{tx}} \|\tilde{\mathbf{h}}_k\|^2} \right)
\leq \log_2\!\left( 1 + \rho_{\mathrm{tx}} \|\mathbf{h}_k^{\perp}\|^2 \right),
\end{equation*}
where the last inequality follows from $(1 + a + b)/(1 + a) \leq 1 + b$ for $a, b \geq 0$. Summing over $k = 1, \ldots, K$ and substituting $\epsilon_k \triangleq \|\mathbf{h}_k^{\perp}\|^2$ gives~\eqref{eq:rate-gap-bound}.

Taking expectations under the second-moment condition $\mathbb E[\mathbf h_k\mathbf h_k^{\mathsf H}]=\mathbf R_k$, we have $\mathbb E[\epsilon_k]=\operatorname{tr}((\mathbf I_M-\overline{\mathbf U}\overline{\mathbf U}^{\mathsf H})\mathbf R_k)$. Define the scheduled-set average covariance $\mathbf R_{\mathcal K}\triangleq K^{-1}\sum_{k=1}^K\mathbf R_k$. Jensen's inequality gives
\begin{equation*}
\mathbb E[R_{\rm full}^{\rm IF}-R_{\rm sub}^{\rm IF}]
\le K\log_2\!\left(1+\rho_{\rm tx}\operatorname{tr}\!\left((\mathbf I_M-\overline{\mathbf U}\overline{\mathbf U}^{\mathsf H})\mathbf R_{\mathcal K}\right)\right).
\end{equation*}
If, additionally, $\mathbf R_{\mathcal K}=\mathbf R$, then the residual term is $\sum_{i=r+1}^M\lambda_i=(1-\eta(r))\operatorname{tr}(\mathbf R)$, yielding
\begin{equation*}
\mathbb E[R_{\rm full}^{\rm IF}-R_{\rm sub}^{\rm IF}]
\le K\log_2\!\left(1+\rho_{\rm tx}(1-\eta(r))\operatorname{tr}(\mathbf R)\right).
\end{equation*}
\hfill $\blacksquare$

\section{Proof of Proposition~\ref{prop:prescreen-loss}}
\label{app:thm1}
For the effective-channel vector $\tilde{\mathbf{h}}_u$ defined in~\eqref{eq:effective-csi-vector}, any unit-norm beam $\mathbf v_u$ gives the single-user rate
\[
R_u(\mathbf v_u)=\log_2\!\big(1+\rho_{\mathrm{tx}}|\tilde{\mathbf h}_u^{\sf H}\mathbf v_u|^2\big)
\le R_u^{\mathrm{MF}}\triangleq\log_2\!\big(1+\rho_{\mathrm{tx}}\|\tilde{\mathbf{h}}_u\|^2\big),
\]
with equality for the matched direction $\mathbf{v}_u = \tilde{\mathbf{h}}_u/\|\tilde{\mathbf{h}}_u\|$ when $\tilde{\mathbf h}_u\neq\mathbf0$; if $\tilde{\mathbf h}_u=\mathbf0$, both sides are zero. Thus $R_u^{\mathrm{MF}}$ is the single-user matched-filter upper bound used in Proposition~\ref{prop:prescreen-loss}.

Because $x\mapsto\log_2(1+\rho_{\mathrm{tx}}x)$ is concave in $x \ge 0$, Jensen's inequality gives
\begin{align*}
  \mathbb{E}\!\big[R_u^{\mathrm{MF}}\big]
  &= \mathbb{E}\!\big[\log_2(1+\rho_{\mathrm{tx}}\|\tilde{\mathbf{h}}_u\|^2)\big] \\
  &\le \log_2\!\big(1+\rho_{\mathrm{tx}}\,\mathbb{E}[\|\tilde{\mathbf{h}}_u\|^2]\big) \\
  &= \log_2\!\big(1+\rho_{\mathrm{tx}}\,s_u\big),
\end{align*}
where, under the analytical second-moment assumption $\mathbb{E}[\mathbf{h}_u\mathbf{h}_u^{\sf H}]=\hat{\mathbf{R}}_u$, $s_u=\mathbb{E}[\|\tilde{\mathbf{h}}_u\|^2]=\mathrm{tr}(\overline{\mathbf{U}}^{\sf H}\hat{\mathbf{R}}_u\overline{\mathbf{U}})$ by~\eqref{eq:prescreen-score}. Under DT mismatch, this bound applies to the channel law represented by $\hat{\mathbf R}_u$ and need not bound the expectation under the physical covariance.

For any user $u \notin \mathcal{S}$, the top-$N_s$ rule ensures $s_u \le s_{(N_s)}$, so
\[
  \mathbb{E}\!\big[R_u^{\mathrm{MF}}\big] \le \log_2\!\big(1+\rho_{\mathrm{tx}}\,s_{(N_s)}\big),
\]
establishing~\eqref{eq:prescreen-bound}. Summing over all $N-N_s$ discarded users yields~\eqref{eq:total-loss}. \hfill$\blacksquare$

\section{Proof of Lemma~\ref{prop:logdet}}
\label{app:lem1}
Let $F(\mathcal{T})=\log\det(\mathbf{I}_r+\gamma\boldsymbol{\Sigma}_{\mathcal{T}})$ with $\boldsymbol{\Sigma}_{\mathcal{T}}=\sum_{u\in\mathcal{T}}\mathbf{C}_u\succeq\mathbf{0}$. Clearly $F(\emptyset)=0$. For $\mathcal{T}\subseteq\mathcal{B}$, $\boldsymbol{\Sigma}_{\mathcal{T}}\preceq\boldsymbol{\Sigma}_{\mathcal{B}}$ gives $F(\mathcal{T})\le F(\mathcal{B})$ by monotonicity of matrix $\log\det$ on the positive-definite cone. Sylvester's determinant identity gives the marginal
\[
\Delta F(u\mid\mathcal{T})=
\log\det\!\left(\mathbf{I}_r+\gamma\mathbf{C}_u^{1/2}
(\mathbf{I}_r+\gamma\boldsymbol{\Sigma}_{\mathcal{T}})^{-1}
\mathbf{C}_u^{1/2}\right).
\]
Since $(\mathbf{I}_r+\gamma\boldsymbol{\Sigma}_{\mathcal{T}})^{-1}\succeq(\mathbf{I}_r+\gamma\boldsymbol{\Sigma}_{\mathcal{B}})^{-1}$, congruence by $\mathbf C_u^{1/2}$ and monotonicity of $\log\det(\mathbf I_r+\gamma\mathbf X)$ for $\mathbf X\succeq0$ yield $\Delta F(u\mid\mathcal{T})\ge\Delta F(u\mid\mathcal{B})$. Thus $F$ is normalized, monotone, and submodular. Applying the standard cardinality-constrained greedy theorem on the fixed ground set $\mathcal V_{\rm cand}$ gives the bound in Lemma~\ref{prop:logdet}~\cite{Nemhauser:1978,Shamaiah:2010}. \hfill$\blacksquare$


\end{document}